\documentclass[12pt,journal,onecolumn]{IEEEtran}
\usepackage{amsfonts,color,morefloats,pslatex}
\usepackage{amssymb,amsthm, amsmath,latexsym}

\newtheorem{theorem}{Theorem}
\newtheorem{lemma}[theorem]{Lemma}

\newtheorem{example}[theorem]{Example}

\newcommand{\tr}{{\mathrm{Tr}}}

\newcommand{\gf}{{\mathrm{GF}}}

\newcommand{\wt}{{\mathtt{wt}}}

\newcommand{\cP}{{\mathcal{P}}}
\newcommand{\cB}{{\mathcal{B}}}

\newcommand{\C}{{\mathcal{C}}}

\newcommand{\bc}{{\mathbf{c}}}

\newcommand{\bD}{{\mathbb{D}}}

\usepackage{blindtext}

\ifCLASSINFOpdf

\else

\fi

\begin{document}

\title{Shortened linear codes from APN and PN functions \thanks{The research of C. Xiang  was supported by the National Natural Science Foundation of
China under grant number 11701187. The research of C. Tang was supported by National Natural Science Foundation
of China under grant number 11871058 and China West Normal University (14E013, CXTD2014-4 and the Meritocracy Research Funds).
The research of C. Ding  was supported by the Hong Kong Research Grants Council, Project No. 16301020. }
}

\author{
Can Xiang, \thanks{C. Xiang is with the College of Mathematics and Informatics, South China Agricultural University, Guangzhou, Guangdong 510642, China (email:cxiangcxiang@hotmail.com). }
Chunming Tang \thanks{C. Tang  is with School of Mathematics  and Information,
China West Normal University, Nanchong, Sichuan 637002, China, and also with  the Department of Computer Science
                                                  and Engineering, The Hong Kong University of Science and Technology,
                                                  Clear Water Bay, Kowloon, Hong Kong, China (email: tangchunmingmath@163.com).}
and Cunsheng Ding \thanks{C. Ding is with the Department of Computer Science
                                                  and Engineering, The Hong Kong University of Science and Technology,
                                                  Clear Water Bay, Kowloon, Hong Kong, China (email: cding@ust.hk).}
}

\maketitle

\begin{abstract}
Linear codes generated by component functions of perfect nonlinear (PN) and almost perfect nonlinear (APN) functions
and the first-order Reed-Muller codes
have been an object of intensive study in coding theory.
The objective of this paper is to investigate some binary shortened codes of two families of linear codes from APN functions
and some $p$-ary shortened codes associated with PN functions.
The weight distributions of these shortened codes and the parameters of their duals are determined.
The parameters of these binary codes and $p$-ary codes are flexible.
Many of the codes presented in this paper are optimal or almost optimal.
The results of this paper show that the shortening technique is very promising for constructing good codes.

\end{abstract}


\IEEEpeerreviewmaketitle

\section{Introduction}\label{Sec-introduct}

Let $\gf(q)$ denote the finite field with $q=p^m$ elements, where $p$ is a prime and $m$ is a positive integer. A $[v,\, k,\,d]$ linear code $\C$ over $\gf(q)$ is a $k$-dimensional subspace of $\gf(q)^v$ with minimum (Hamming) distance $d$.
Let $A_i$ denote the number of codewords with Hamming weight $i$ in a code
$\C$ of length $v$. The weight enumerator of $\C$ is defined by
$
1+A_1z+A_2z^2+ \cdots + A_v z^v.
$
The sequence $(1,A_1,\ldots,A_v)$ is called the weight distribution of $\C$ and is an important research topic in coding theory,
as it contains crucial information about the error correcting capability of the code. Thus the study of the weight distribution
has attracted much attention in coding theory and much work focuses on the determination of
the weight distributions of linear codes (see, for example, \cite{ding2018,Ding16,DingDing2,Ding15,sihem2020,sihem2017,Tangit2016,TXF2017,zhou20132,zhou20131}).
Denote by $\C^\bot$  and $(A_0^{\perp}, A_1^{\perp}, \dots, A_\nu^{\perp})$ the dual code of a linear code $\C$ and its weight distribution, respectively.
The \emph{Pless power moments} \cite{HP10}, i.e.,
\begin{align}\label{eq:PPM}
 \sum_{i=0}^\nu  i^t A_i= \sum_{i=0}^t (-1)^i A_i^{\perp}   \left [  \sum_{j=i}^t   j ! S(t,j) q^{k-j} (q-1)^{j-i} \binom{\nu-i}{\nu -j}  \right ],
 \end{align}
play an important role in calculating
the weight distributions of linear codes, where $A_0=1$, $0\le t \le \nu$ and $S(t,j)=\frac{1}{j!} \sum_{i=0}^j  (-1)^{j-i} \binom{j}{i} i^t$.
A code $\C$ is said to be a $t$-weight code  if the number of nonzero
$A_i$ in the sequence $(A_1, A_2, \cdots, A_v)$ is equal to $t$.  A  $[v,k,d]$ code over $\gf(q)$ is said to be \emph{distance-optimal} if no $[v,k,d']$ code over $\gf(q)$ with $d'>d$ exists, \emph{dimension-optimal} if no $[v,k',d]$ code over $\gf(q)$ with $k'>k$ exists, and \emph{length-optimal} if no $[v',k,d]$ code over $\gf(q)$ with
$v'< v$ exists. A linear code is said to be optimal if it is distance-optimal, or dimension-optimal, or length-optimal, or meets a bound
for linear codes.

Let $\C$ be a $[\nu,k,d]$ linear code over $\gf(q)$ and $T$ a set of $t$ coordinate positions in $\C$.
We use $\mathcal  C^T$ to denote the code obtained by puncturing $\mathcal  C$  on
$T$, which is called the  \emph{punctured code}  of $\mathcal C$ on $T$.
Let $\mathcal C(T)$ be the set of codewords of $\mathcal C$  which are
$0$ on $T$.
We now puncture $\mathcal C(T)$ on $T$, and obtain a linear code $\mathcal C_{T}$, which is called the \emph{shortened code} of $\mathcal C$ on $T$.
The following lemma plays an important role in determining the parameters of the punctured and shortened codes of $\mathcal{C}$.

\begin{lemma}\cite[Theorem 1.5.7]{HP10}\label{lem:C-S-P}
Let $\C$ be a $[\nu,k,d]$ linear code over $\gf(q)$ and  $d^{\perp}$  the minimum distance of $\mathcal  C^{\perp}$.
Let $T$ be any set of $t$ coordinate positions. Then the following hold:
\begin{itemize}
  \item $\left ( \mathcal C_{T} \right )^{\perp} = \left ( \mathcal C^{\perp}  \right)^T$ and $\left ( \mathcal C^{T} \right )^{\perp} = \left ( \mathcal C^{\perp}  \right)_T$.
  \item If $t<\min \{d, d^{\perp} \}$, then the codes $\mathcal C_{T}$ and $\mathcal C^T$ have  dimension  $k-t$  and $k$, respectively.
\end{itemize}
\end{lemma}

The shortening and puncturing techniques are two important approaches to constructing new linear codes. Very recently, Tang et al. obtained some ternary linear codes with few weights by shortening and puncturing a class of ternary codes in \cite{Tangdcc2019}. Afterwards, they  presented a general theory for punctured and shortened codes
of linear codes supporting t-designs and generalized the Assmus-Mattson theorem in \cite{Tangit2019}.
Liu et al. studied some shortened linear codes over finite fields in \cite{LDT20}.
However, till now not much work about shortened codes has been done and
it is in general hard to determine the weight distributions of shortened codes. Motivated by these facts,
we investigate some shortened codes of linear codes from  almost perfect nonlinear (APN) and  perfect nonlinear (PN) functions,  and determine their parameters in this paper. Many of these shortened codes are optimal or almost optimal.

The rest of this paper is arranged as follows. Section \ref{sec-pre} introduces some notation and results
related to group characters, Gauss sums, $t$-designs and linear codes from APN and PN functions. Section \ref{sec-general} gives some general results about shortened codes. Section \ref{sec-APN} investigates
some shortened codes of binary linear codes from APN functions. Section \ref{sec-PN} studies some shortened codes of two classes of special linear codes from PN functions. Section  \ref{sec-summary}
concludes this paper and makes concluding remarks.

\section{Preliminaries}\label{sec-pre}

In this section, we briefly recall some results on group characters, Gauss sums, $t$-designs, and linear codes from APN and PN functions. These results will be used later in this paper. We begin this section by fixing some notation
throughout this paper.

\begin{itemize}
\item $p$ is a prime and $p^*=(-1)^{(p-1)/2}p$ for odd prime $p$.
  \item $\zeta_p=e^{\frac{2\pi \sqrt{-1}}{p}}$ is the primitive $p$-th
root of unity.
\item $q$ is a power of $p$.
\item $\gf(q)^*=\gf(q)\setminus \{0\}$.
\item $\tr_{q/p}$ is the trace function from $\gf(q)$ to $\gf(p)$.
\item $\textup{SQ}$ and $\textup{N\textup{SQ}}$ denote the set of all  squares and nonsquares in $\gf(p)^{*}$, respectively.
\item $\eta$ and $\bar{\eta}$ are the quadratic characters of $\gf(q)^{*}$ and  $\gf(p)^{*}$, repsectively. We extend these quadratic characters
by letting $\eta(0)=0$ and $\bar{\eta}(0)=0$.
\end{itemize}

\subsection{Group characters and Gauss sums}

An additive character of $\gf(q)$ is a nonzero function $\chi$
from $\gf(q)$ to the set of nonzero complex numbers such that
$\chi(x+y)=\chi(x) \chi(y)$ for any pair $(x, y) \in \gf(q)^2$.
For each $b\in \gf(q)$, the function
\begin{eqnarray}\label{dfn-add}
\chi_b(x)=\zeta_p^{\tr_{q/p}(bx)}
\end{eqnarray}
defines an additive character of $\gf(q)$. When $b=0$,
$\chi_0(x)=1 \mbox{ for all } x\in\gf(q),
$
and $\chi_0$ is called the {\em trivial additive character} of
$\gf(q)$. The character $\chi_1$ in (\ref{dfn-add}) is called the
{\em canonical additive character} of $\gf(q)$.
It is well known that every additive character of $\gf(q)$ can be
written as $\chi_b(x)=\chi_1(bx)$ \cite[Theorem 5.7]{LN}. The orthogonality  relation of additive characters is given by
$$\sum_{x\in \gf(q)}\chi_1(ax)=\left\{
\begin{array}{rl}
q    &   \mbox{ for }a=0,\\
0    &   \mbox{ for }a\in \gf(q)^*.
\end{array} \right. $$

The Gauss sum $G(\eta, \chi_1)$ over $\gf(q)$ is defined by
\begin{eqnarray}
G(\eta, \chi_1)=\sum_{x \in \gf(q)^*} \eta(x) \chi_1(x) = \sum_{x \in \gf(q)} \eta(x) \chi_1(x)
\end{eqnarray}
and
the Gauss sum $G(\bar{\eta}, \bar{\chi}_1)$ over $\gf(p)$ is defined by
\begin{eqnarray}
G(\bar{\eta}, \bar{\chi}_1)=\sum_{x \in \gf(p)^*} \bar{\eta}(x) \bar{\chi}_1(x)
= \sum_{x \in \gf(p)} \bar{\eta}(x) \bar{\chi}_1(x),
\end{eqnarray}
where $\bar{\chi}_1$ is the canonical additive character of $\gf(p)$.

The following four lemmas are proved in \cite[Theorems 5.15, 5.33, Corollary 5.35]{LN} and \cite[Lemma 7]{DingDing2}, respectively.

\begin{lemma}\cite{LN} \label{lem-32A1}
Let $q=p^m$ and $p$ be an odd prime. Then
\begin{eqnarray*}G(\eta,\chi_1)&=&(-1)^{m-1}(\sqrt{-1})^{(\frac{p-1}{2})^2m}\sqrt{q}\\
 &=&\left\{
\begin{array}{lll}
(-1)^{m-1}\sqrt{q}    &   \mbox{ for }p\equiv 1\pmod{4},\\
(-1)^{m-1}(\sqrt{-1})^{m}\sqrt{q}    &   \mbox{ for }p\equiv 3\pmod{4}.
\end{array} \right.
\end{eqnarray*}
and
$$
G(\bar{\eta}, \bar{\chi}_1)= \sqrt{-1}^{(\frac{p-1}{2})^2 } \sqrt{p}=\sqrt{p*}.
$$
\end{lemma}

\begin{lemma}\cite{LN} \label{lem-32A2}
Let $\chi$ be a nontrivial additive character of $\gf(q)$ with $q$ odd, and let
$f(x)=a_2x^2+a_1x+a_0 \in \gf(q)[x]$ with $a_2 \ne 0$. Then
$$
\sum_{x \in \gf(q)} \chi(f(x)) = \chi(a_0-a_1^2(4a_2)^{-1}) \eta(a_2) G(\eta, \chi).
$$
\end{lemma}

\begin{lemma}\cite{LN} \label{lem-charactersum-evenq}
Let $\chi_b$ be a nontrivial additive character of $\gf(q)$ with $q$ even and  $f(x)=a_2x^2+a_1x+a_0\in \gf(q)[x]$, where $b\in \gf(q)^*$.  Then
$$\sum_{x\in \gf(q)}\chi_b(f(x))=\left\{\begin{array}{ll}
\chi_b(a_0)q    &   \mbox{ if }a_2=ba_{1}^{2},\\
0    &   \mbox{ otherwise. }
\end{array} \right.$$
\end{lemma}

\begin{lemma} \cite{DingDing2} \label{lem-bothcharac}
Let $p$ be an odd prime.
If $m \ge 2$ is even, then $\eta(x)=1$ for each $x \in \gf(p)^*$.
If $m \ge 1$ is odd, then  $\eta(x)=\bar{\eta}(x)$ for each $x \in \gf(p)$.
\end{lemma}

Let $e$ be a positive integer and $(a, b)\in \gf(q)^2$, define the exponential sum
\begin{equation}\label{def S}
S_e (a, b)=\sum\limits_{x\in \gf(q)}\chi_1\left(a
x^{p^{e}+1}+b x\right).
\end{equation}
Then we have the following five known results.

\begin{lemma}\cite{Coult}\label{expsum}
Let $e$ be a positive integer and $m$ be even with $\gcd(m,e)=1$. Let $p=2$, $q=2^m$ and $a\in \gf(q)^*$. Then
\begin{eqnarray*}
S_e(a,0)=\left\{
\begin{array}{l}
(-1)^{\frac{m}{2}} 2^{\frac{m}{2}}       ~~~~~~\mbox{ if $a \ne \alpha^{3t}$ for any $t$,} \\
-(-1)^{\frac{m}{2}} 2^{\frac{m}{2}+1}      ~\mbox{ if $a = \alpha^{3t}$ for some $t$,}
\end{array}
\right.
\end{eqnarray*}
where $\alpha$ is a generator of $\gf(q)^*$.
\end{lemma}

\begin{lemma}\cite{Cao2016}\label{expsum1}
Let $e,h$ be positive integers and $m$ be even with $\gcd(m,e)=1$. Let $p=2$, $q=2^m$ and $a\in \gf(q)^*$. Then
\begin{eqnarray*}
\sum_{b\in \gf(q)^*} \left( S_e(a,b)\right)^h =\left\{
\begin{array}{ll}
(2^m-1) 2^{\frac{m}{2}\cdot h}               &   \mbox{ if $h$ is even and $a \ne \alpha^{3t}$ for any $t$,} \\
(2^{m-2}-1) 2^{(\frac{m}{2}+1) \cdot h}      &   \mbox{ if $h$ is even and $a = \alpha^{3t}$ for some $t$,}
\end{array}
\right.
\end{eqnarray*}
where $\alpha$ is a generator of $\gf(q)^*$.
\end{lemma}

\begin{lemma}\cite{coulter2002} \label{lem-psum}
Let $p$ be an odd prime, $q=p^m$, and $e$ be any positive integer such that $m/\gcd(m,e)$ is odd. Suppose $a\in \gf(q)^* $ and $b \in \gf(q)^*$. Let $x_{a,b}$ be the unique solution of the equation
$$
a^{p^e}x^{p^{2e}}+a x+ b^{p^e}=0.
$$
Then
\begin{align*}
S_e (a, b)=
\left\{
  \begin{array}{ll}
    (-1)^{m-1} \sqrt{q} \eta(-a) \chi_1(-a x_{a,b}^{p^e+1} ),        & \mbox{if}~ p\equiv 1 ~mod~4, \\
    (-1)^{m-1} \sqrt{-1}^{3m} \sqrt{q} \eta(-a) \chi_1(-a x_{a,b}^{p^e+1} ), & \mbox{if}~ p\equiv 3 ~mod~4. \\
  \end{array}
\right.
\end{align*}
\end{lemma}

\begin{lemma}\cite{YCD2006}\label{lem-sumxp}
Let the notation and assumptions be the same as those in the previous lemma. Write $\Delta=\sum_{c\in \gf(p)^*}S_e (\ ac,\ bc )$.
Then we have the following results.
\begin{itemize}
  \item If $m$ is odd, then
  \begin{align*}
\Delta=
\left\{
  \begin{array}{ll}
    0,        & \mbox{if $\tr_{q/p}(a(x_{a,b})^{p^e+1})=0$}, \\
    \eta(a)\eta( \tr_{q/p}(a(x_{a,b})^{p^e+1}) )\sqrt{q} \sqrt{p^*}, & \mbox{if $p\equiv 1 ~mod~4$ and $\tr_{q/p}(a(x_{a,b})^{p^e+1})\neq 0$ }, \\
    \eta(a)\eta( \tr_{q/p}(a(x_{a,b})^{p^e+1}) )\sqrt{-1}^{3m}\sqrt{q} \sqrt{p^*}, & \mbox{if $p\equiv 3 ~mod~4$ and $\tr_{q/p}(a(x_{a,b})^{p^e+1})\neq 0$ }.
  \end{array}
\right.
\end{align*}

  \item  If $m$ is even, then
  \begin{align*}
\Delta=
\left\{
  \begin{array}{ll}
    -(p-1)\eta(a)\sqrt{q},        & \mbox{if $p\equiv 1 ~mod~4$ and $\tr_{q/p}(a(x_{a,b})^{p^e+1})=0$}, \\
    \eta(a)\sqrt{q},        & \mbox{if $p\equiv 1 ~mod~4$ and $\tr_{q/p}(a(x_{a,b})^{p^e+1})\neq 0$}, \\
    -\sqrt{-1}^m (p-1)\eta(a)\sqrt{q},        & \mbox{if $p\equiv 3 ~mod~4$ and $\tr_{q/p}(a(x_{a,b})^{p^e+1})=0$}, \\
    \sqrt{-1}^m \eta(a)\sqrt{q},        & \mbox{if $p\equiv 3 ~mod~4$ and $\tr_{q/p}(a(x_{a,b})^{p^e+1})\neq 0$}. \\
  \end{array}
\right.
\end{align*}
\end{itemize}
\end{lemma}

\begin{lemma}\cite{YCD2006}\label{lem-NN0}
Let $p$ be an odd prime, $m$ and $e$ be positive integers such that $m/\gcd(m,e)$ is odd. Let $q=p^m$. Define
\begin{eqnarray*}
\hat{N}_0(a,b)=\sharp\{x\in \gf(q):\tr_{q/p}(ax^{p^e+1}+bx)=0\}.
\end{eqnarray*}
Then we have the following results.
\begin{itemize}
  \item If $a=0$ and $b=0$, then $\hat{N}_0(a,b)=q$.
  \item If $a=0$ and $b\neq 0$, then $\hat{N}_0(a,b)=p^{m-1}$.
  \item If $a\neq 0$ and $b=0$, then
\begin{align*}
 \hat{N}_0(a,b)=
\left\{
  \begin{array}{ll}
     p^{m-1}                                                             &  \mbox{if $m$ is odd}, \\
     \frac{1}{p} \left ( q-(p-1)\eta(a)\sqrt{q} \right),                               & \mbox{if $m$ is even and $p\equiv 1 ~mod~4$}, \\
      \frac{1}{p}\left( q-\sqrt{-1}^m (p-1)\eta(a)\sqrt{q} \right),        & \mbox{if $m$ is even and $p\equiv 3 ~mod~4$.}
  \end{array}
\right.
\end{align*}
  \item If $a\neq 0$ and $b\neq 0$, then $\hat{N}_0(a,b)=\frac{1}{p}(q+\Delta)$, where $\Delta$ was given in Lemma \ref{lem-sumxp}.
\end{itemize}
\end{lemma}

\subsection{Combinatorial $t$-designs and related results}

Let $k$, $t$ and $v$ be positive integers with $1\leq t \leq k \leq v$. Let $\cP$ be a set of $v \ge 1$  elements, and let $\cB$ be a set of $k$-subsets of $\cP$. The \emph{incidence structure} $\bD = (\cP, \cB)$ is said to be a $t$-$(v, k, \lambda)$ {\em design\index{design}} if every $t$-subset of $\cP$ is contained in exactly $\lambda$ elements of $\cB$.
The elements of $\cP$ are called points, and those of $\cB$ are referred to as blocks.
We usually use $b$ to denote the number of blocks in $\cB$.  A $t$-design is called {\em simple\index{simple}} if $\cB$ has no repeated blocks. A $t$-design is called symmetric if $v = b$ and trivial if $k = t$ or $k = v$.
When  $t \geq 2$ and $\lambda=1$, a $t$-design is called a {\em Steiner system\index{Steiner system}} and traditionally denoted by $S(t,k, v)$.

Linear codes and $t$-designs are companions.
A $t$-design $\mathbb  D=(\mathcal P, \mathcal B)$ induces a linear code over GF($p$) for
any prime $p$.
Let $\mathcal P=\{p_1, \dots, p_{\nu}\}$. For any block $B\in \mathcal B$,   the \emph{characteristic vector} of $B$ is defined
by the vector $\mathbf{c}_{B}
=(c_1, \dots, c_{\nu})\in \{0,1\}^{\nu}$, where
\begin{align*}
c_i=
\left\{
  \begin{array}{ll}
    1, & \text{if}~ p_i \in B, \\
    0, & \text{if}~ p_i \not \in B.
  \end{array}
\right.
\end{align*}
For a prime $p$, a \emph{linear code} $\mathsf{C}_{p}(\mathbb D)$ over the prime field $\mathrm{GF}(p)$ from the design $\mathbb D$ is spanned by the characteristic vectors of the blocks of $\mathbb B$, which is
 the subspace $\mathrm{Span}\{\mathbf{c}_{B}: B\in \mathcal B\}$ of the vector space $\mathrm{GF}(p)^{\nu}$.
Linear codes $\mathsf{C}_{p}(\mathbb D)$ from designs $\mathbb D$ have been studied and documented in the literature
 (see, for example,
\cite{AK92,Ding15,Ton98,Ton07}).

On the other hand,  a linear code $\C$ may induce a $t$-design under certain conditions, which is formed by
 supports of  codewords of a fixed Hamming weight in $\C$.
 Let
$\mathcal P(\mathcal C)$  be the set of the coordinate positions of $\C$, where $\# \mathcal P(\mathcal C)=v$ is the length of $\mathcal C$.
For a codeword $\mathbf c =(c_i)_{i\in \mathcal P(\mathcal C)}$ in $\C$, the \emph{support} of  $\mathbf c$
is defined by
\begin{align*}
\mathrm{Supp}(\mathbf c) = \{i: c_i \neq 0, i \in \mathcal P(\mathcal C)\}.
\end{align*}
Let $\mathcal B_{w}(\mathcal C)
=\{\mathrm{Supp}(\mathbf c): wt(\mathbf{c})=w
~\text{and}~\mathbf{c}\in \mathcal{C}\}$. For some special $\mathcal C$, $\left (\mathcal P(\mathcal C),  \mathcal B_{w}(\mathcal C) \right)$
is a $t$-design.
In this way, many $t$-designs are derived  from linear codes (see, for example, \cite{AK92,Ding18dcc,Ding18jcd,DLX17,HKM04,HMT05,KM00,MT04,Ton98,Ton07}). A major approach to constructing $t$-designs from linear codes is the use of linear codes with $t$-homogeneous or $t$-transitive automorphism groups (see \cite[Theorem 4.18]{ding2018}).
Another major approach to constructing $t$-designs from codes is the use of the
Assmus-Mattson Theorem \cite{AM74,HP10}. The following Assmus-Mattson Theorem for constructing simple $t$-designs was developed in  \cite{AM69}.

\begin{theorem}\label{thm-AMTheorem}
Let $\mathcal C$ be a linear code over $\mathrm{GF}(q)$ with length $\nu$ and minimum weight $d$.
Let $d^{\perp}$ denote the minimum weight of the dual code  $\mathcal C^{\perp}$ of $\mathcal C$. Let  $t~(1\le t <\min \{d, d^{\perp}\})$  be an integer such that there are at most $d^{\perp}-t$ weights of $\mathcal C$ in the range $\{1,2, \ldots, \nu-t\}$.
Then the following hold:
\begin{itemize}
\item $(\mathcal P(\mathcal C), \mathcal B_{k}(\mathcal C))$ is a simple $t$-design provided that
      $A_k \neq 0$ and $d \leq k \leq w$, where $w$ is defined to be the largest integer
      satisfying $w \leq \nu$ and
      $$
      w-\left\lfloor \frac{w+q-2}{q-1} \right\rfloor <d.
      $$
\item $(\mathcal P(\mathcal C^{\perp}), \mathcal B_{k}(\mathcal C^{\perp}))$ is a simple $t$-design provided that
      $A_k^\perp  \neq 0$ and $d^\perp \leq k \leq w^\perp$, where $w^\perp$ is defined to be the largest integer
      satisfying $w^\perp \leq \nu$ and
      $$
      w^\perp-\left\lfloor \frac{w^\perp+q-2}{q-1} \right\rfloor <d^\perp.
      $$
\end{itemize}
\end{theorem}

We will need the following results about the punctured and shortened codes of $\C$ documented  in  \cite[Lemma 3.1,Theorem 3.2]{Tangit2019} .

\begin{lemma}\cite{Tangit2019}\label{lem:P:k:k+t}
Let $\mathcal C$ be a  linear code of length $\nu$ and minimum distance $d$ over $\mathrm{GF}(q)$ and  $d^{\perp}$  the minimum distance of $\mathcal  C^{\perp}$.
 Let  $t$ and $k$ be two positive integers with  $0< t <\min \{d, d^{\perp}\}$ and $1\le k\le  \nu-t$.
Let $T$ be  a  set of $t$ coordinate positions in $\mathcal  C$.
Suppose that $\left ( \mathcal P(\mathcal C) , \mathcal B_i(\mathcal C) \right )$ is a $t$-design for all $i$ with $k\le i \le k+t$.
Then
$$A_k(\mathcal C^T) =\sum_{i=0}^t  \frac{\binom{\nu-t}{k} \binom{k+i}{t} \binom{t}{i} }{\binom{\nu-t}{k-t+i} \binom{\nu}{t}} A_{k+i}(\mathcal C).$$
\end{lemma}

\begin{theorem}\cite{Tangit2019} \label{thm:sct-code}
Let $\mathcal C$ be a $[\nu, \bar{k}, d]$ linear code  over $\mathrm{GF}(q)$ and  $d^{\perp}$ be the minimum distance of $\mathcal  C^{\perp}$.
 Let  $t$ be a positive integer with  $0< t <\min \{d, d^{\perp}\}$.
Let $T$ be  a  set of $t$ coordinate positions in $\mathcal  C$.
Suppose that $\left ( \mathcal P(\mathcal C) , \mathcal B_i(\mathcal C) \right )$ is a $t$-design for any $i$ with $d \le i \le \nu-t$.
Then the shortened code $\mathcal C_T$ is a linear code of length $\nu-t$ and dimension $\bar{k}-t$. The weight distribution
$\left ( A_k(\mathcal C_T) \right )_{k=0}^{\nu-t}$ of $\mathcal C_T$ is independent of the specific choice of the elements
in $T$. Specifically,
$$A_k(\mathcal C_T) =\frac{ \binom{k}{t} \binom{\nu-t}{k}}{ \binom{\nu }{t} \binom{\nu-t}{k-t}}A_k(\mathcal C).$$
\end{theorem}


\subsection{Linear codes from APN and PN functions}

Let $m, \tilde{m}$ be two positive integers with $m \geq  \tilde{m}$  and $F$ be a mapping from $\gf(p^{m})$ to  $\gf(p^{\tilde{m}})$. Define
\[\delta_F=\max\{\delta_F(a,b): a \in \gf(p^{m})^*, b \in \gf(p^{\tilde{m}})\},\]
where $\delta_F(a,b)=\#\{x \in \gf(p^{m}) : F(x+a)-F(x)=b\}$, $a\in\gf(p^{m})$ and $b\in \gf(p^{\tilde{m}}) $.
The function $F(x)$ is called PN function  if $\delta_F=p^{m-\tilde{m}}$, and it is called APN function  if $m=\tilde{m}$ and $\delta_F=2$.
From the above definition one immediately sees that $F(x)$ is PN if and only if $F(x+a)-F(x)$ is balanced for each $a\in \gf(p^{m})^{*}$. Currently, all known PN and APN functions over $\gf(p^{m})$ are summarized in  \cite{Budaghyan-Helleseth,car2016,car2015,Coulter-Henderson-Hu,Coulter-Matthews,Dembowski-Ostrom,ding2018,Ding-Yuan,Zha-Kyureghyan-Wang}.
It is known that PN and APN functions are very important functions for constructing linear codes with good parameters  (see, for
example, \cite{YCD2005,sihem2019,WLZ2020,YCD2006}).

Let $q=p^m$ and let $\C$ denote the linear code of length $q$ defined by
\begin{align}\label{eq:cf}
\C=\left\{\left( \tr_{q/p}(af(x)+bx+c) \right)_{x \in \gf(q)}:a,b,c\in \gf(q) \right\},
\end{align}
where $f(x)$ is a polynomial over $\gf(q)$.
Then we can regard $\gf(q)$ as  the set of the coordinate positions $\mathcal{P}(\mathcal C)$ of $\mathcal C$.
It is known that $\C$ has dimension $2m+1$  and the weight distribution in Table \ref{tab-cf} when $p=2$, $m\geq 5$ is odd and $f(x)=x^{s}$ is an APN function, where $s$ takes
 the following values \cite{ding2018}.

\begin{itemize}
  \item $s=2^e+1$, where $\gcd(e,m) = 1$ and $e$ is a positive integer.
  \item $s=2^{2e}-2^e+1$, where $e$ is a positive integer and $\gcd(e,m) = 1$.
  \item $s=2^{(m-1)/2}+3$.
  \item $s=2^{(m-1)/2}+2^{(m-1)/4}-1$, where $m \equiv 1 ~(~mod ~4~)$.
  \item $s=2^{(m-1)/2}+2^{(3m-1)/4}-1$, where $m \equiv 3 ~(~mod~4~)$.
\end{itemize}

When $f(x)=x^{2^e+1}$, $p=2$ and $m\geq 4$ is even with $\gcd(e,m) = 1$,  the code $\C$  defined in (\ref{eq:cf}) has dimension $2m+1$  and the weight distribution in Table \ref{tab-cf1} \cite{ding2018}.

\begin{table}[ht]
\begin{center}
\caption{The weight distribution of $\C$ for $m$ odd }\label{tab-cf}
\begin{tabular}{cc} \hline
Weight  &  Multiplicity   \\ \hline
$0$          &  $1$ \\
$2^{m-1}-2^{(m-1)/2}$    & $ (2^m-1)2^{m-1}$ \\
$2^{m-1}$    & $ (2^m-1)(2^{m+1}-2^{m}+2)$ \\
$2^{m-1}+2^{(m-1)/2}$    & $ (2^m-1)2^{m-1}$ \\
$2^{m}$  & $1$ \\
\hline
\end{tabular}
\end{center}
\end{table}

\begin{table}[ht]
\begin{center}
\caption{The weight distribution of $\C$ for $m$ even }\label{tab-cf1}
\begin{tabular}{cc} \hline
Weight  &  Multiplicity   \\ \hline
$0$          &  $1$ \\
$2^{m-1}-2^{m/2}$    & $ (2^m-1)2^{m-2}/3$ \\
$2^{m-1}-2^{(m-2)/2}$    & $ (2^m-1)2^{m+1}/3$ \\
$2^{m-1}$    & $ 2(2^m-1)(2^{m-2}+1)$ \\
$2^{m-1}+2^{(m-2)/2}$    & $ (2^m-1)2^{m+1}/3$ \\
$2^{m-1}+2^{m/2}$    & $ (2^m-1)2^{m-2}/3$ \\
$2^{m}$  & $1$ \\
\hline
\end{tabular}
\end{center}
\end{table}

It is known that the code $\C$ defined in (\ref{eq:cf}) has dimension $2m+1$ and a few weights when $p$ is an odd prime and $f(x)=x^{s}$ is a PN function.
If $s$ takes
  the following values \cite{ding2018,LQ2009}
\begin{itemize}
  \item $s=2$,
  \item $s=p^e+1$, where $m/\gcd(m,e)$ is odd,
  \item $s=(3^e+1)/2$, where $p=3$, $e$ is odd and $\gcd(m,e)=1$,
\end{itemize}
then $f(x)=x^s$ is a PN and also planar function, $\mathrm{Tr}_{q/p}(\beta f(x)))$ is a weakly regular bent function \cite{GHHK2012,HHKWX2009} for any $\beta\in
\gf(q)^*$ ,  and the code $\C$ defined in (\ref{eq:cf}) has four or six weights  \cite{LQ2009}.

Let $f(x)$ be a function from $\gf(q)$
to $\gf(p)$, the  Walsh transform of $f$ at a point $\beta \in \gf(q)$ is defined by
$$
\mathcal{W}_{f}(\beta)=
\sum_{x\in\gf(q)}\zeta_p^{f(x)
-\tr_{q/p}(\beta x)}.
$$
The function $f(x)$ is said to be a  $p$-ary bent function, if $|\mathcal{W}_f(\beta)|=p^{\frac{m}{2}}$ for any $\beta\in \mathbb{F}_q$.
A  bent function $f(x)$ is weakly regular if
there exists a complex $u$ with unit magnitude
satisfying
$\mathcal{W}_f(\beta)=up^{\frac{m}{2}}\zeta_p^{f^*(\beta)}$
for some function $f^*(x)$.
Such function $f^*(x)$ is called the dual of
$f(x)$. A weakly regular bent function
$f(x)$ satisfies
\begin{equation*}
\mathcal{W}_f(\beta)=\varepsilon \sqrt{p^*}^{m} \zeta_p^{f^*(\beta)},
\end{equation*}
where $\varepsilon =\pm 1$ is called the sign of
the Walsh Transform of $f(x)$.
Let $\mathcal{RF}$ be the set of
$p$-ary weakly regular bent functions with the following two properties:
\begin{itemize}
  \item $f(0)=0$; and
  \item $f(ax)=a^hf(x)$ for any $a\in \gf(p)^*$ and $x\in \gf(q)$, where $h$ is a positive even integer
with $\gcd(h-1,p-1)=1$.
\end{itemize}

We will need the following results about $p$-ary weakly regular bent functions in  \cite{Tangit2016}.

\begin{lemma}\cite{Tangit2016} \label{2lem11}
Let $\beta\in \gf(q)^*$ and
$f(x)\in \mathcal{RF}$ with   $\mathcal{W}_f(0)=\varepsilon\sqrt{p^*}^{m}$.  Define
$$
N_{f,\beta}=\#\{x\in \gf(q):
f(x)=0 ~\textrm{and}~ \tr_{q/p}(\beta x)=0\}.
$$
If $f^*(\beta)=0$, then
$$
N_{f,\beta}=
\left\{
  \begin{array}{ll}
    p^{m-2}+\varepsilon \bar{\eta}^{m/2}(-1)(p-1)p^{(m-2)/2}, & \hbox{if $m$ is even;} \\
    p^{m-2},                                                 & \hbox{if $m$ is odd.}
  \end{array}
\right.
$$
\end{lemma}

\begin{lemma} \cite{Tangit2016} \label{2lem14}
Let $\beta\in \gf(q)^*$ and
$f(x)\in \mathcal{RF}$ with   $\mathcal{W}_f(0)=\varepsilon\sqrt{p^*}^{m}$. Let
$$
N_{sq,\beta}=\#\{x\in \gf(q):
f(x)\in  \textrm{SQ} ~\textrm{and}~ \tr_{q/p}(
\beta x)=0 \},
$$
and
$$
N_{nsq,\beta}=\#\{x\in \gf(q):
f(x)\in \textrm{N\textrm{SQ}} ~\textrm{and}~ \tr_{q/p}(
\beta x)=0 \}.
$$
We have the following results.
\begin{itemize}
  \item If $m$ is even and $f^*(\beta)=0$, then
$$
N_{sq,\beta}=N_{nsq,\beta}=\frac{p-1}{2}\left(p^{m-2}-\varepsilon \bar{\eta}^{m/2}(-1)p^{(m-2)/2}\right).
$$
  \item  If $m$ is odd and $f^*(\beta)=0$, then
$$
N_{sq,\beta}=\frac{p-1}{2}\left( p^{m-2}+\varepsilon \sqrt{p*}^{m-1} \right)
$$
and
$$
N_{nsq,\beta}= \frac{p-1}{2}\left( p^{m-2}-\varepsilon \sqrt{p*}^{m-1} \right).
$$
\end{itemize}

\end{lemma}

\section{Shortened binary linear codes with special weight distributions}\label{sec-general}

In this section, we give some general results on  the shortened codes of linear codes with the weight distributions in Tables \ref{tab-cf} and \ref{tab-cf1}.

Let $T$  be a set of $t$ coordinate positions in $\C$ (i.e., $T$ is a $t$-subset of $\mathcal P(\C)$). Define
$$\Lambda _{T,w}(\C)=\{\mathrm{Supp}(\mathbf c):~\mathbf{c}\in \C,~wt(\mathbf{c})=w, ~and~T\subseteq \mathrm{Supp}(\mathbf c) \}. $$
and $\lambda _{T,w}(\C) = \# \Lambda _{T,w}(\C)$.

We will consider some shortened code $\C_T$ of $\C$ for the case $m \geq 4$ and $t \geq 1$ .

\subsection{Shortened linear codes holding $t$-designs}

Let $p=2$ and $q=2^m$. Notice that if a binary code $\C$ has length $2^m$  and the weight distribution in Table \ref{tab-cf} (resp. Table \ref{tab-cf1}), then the code $\C$ holds $3$-designs (resp. $2$-designs) (\cite{ding2018,DT2020ccds}). The following two theorems are easily derived from Theorem \ref{thm:sct-code}, Tables \ref{tab-cf} and \ref{tab-cf1}, and we omit their proofs.

\begin{theorem}\label{main-design1}
Let $m\geq 5$ be odd, and $\C$ be a binary linear code with length $2^m$  and the weight distribution in Table \ref{tab-cf}. Let  $T$ be a $t$-subset of $\mathcal P(\C)$.
We have the following results.
\begin{itemize}
  \item   If $t=1$, then the shortened code $\C_{T}$ is a $[2^{m}-1, 2m,   2^{m-1} - 2^{(m -1)/2}]$ binary linear code with the weight distribution in Table \ref{tab-des11}.
  \item   If $t=2$, then the shortened code $\C_{T}$ is a $[2^{m}-2, 2m-1, 2^{m-1} - 2^{(m -1)/2}]$ binary linear code with the weight distribution in Table \ref{tab-des12}.
  \item   If $t=3$, then the shortened code $\C_{T}$ is a $[2^{m}-3, 2m-2, 2^{m-1} - 2^{(m -1)/2}]$ binary linear code with the weight distribution in Table \ref{tab-des13}.
\end{itemize}
\end{theorem}
\begin{table}[ht]
\begin{center}
\caption{The weight distribution of $\C_T$ for $m$ odd and $t =1$ }\label{tab-des11}
\begin{tabular}{cc} \hline
Weight  &  Multiplicity   \\ \hline
$0$          &  $1$ \\  [2mm]
$2^{m-1}-2^{(m-1)/2}$    & $ 2^{(m-5)/2} (2^m-1) (2 + 2^{(1 + m)/2})$  \\ [2mm]
$2^{m-1}$    & $-1 + 2^{m-1} + 2^{2m-1}$ \\ [2mm]
$2^{m-1}+2^{(m-1)/2}$    & $ 2^{(m-5)/2} (2^m-1) (-2 + 2^{(1 + m)/2})$ \\ [2mm]
\hline
\end{tabular}
\end{center}
\end{table}

\begin{table}[ht]
\begin{center}
\caption{The weight distribution of $\C_T$ for $m$ odd and $t =2$ }\label{tab-des12}
\begin{tabular}{cc} \hline
Weight  &  Multiplicity   \\ \hline
$0$          &  $1$ \\  [2mm]
$2^{m-1}-2^{(m-1)/2}$    & $ 2^{(m-7)/2} (-4 + 2^{2 + m} + 2^{(1 + 3 m)/2})$  \\ [2mm]
$2^{m-1}$    & $-1 +2^{2m-2} $ \\ [2mm]
$2^{m-1}+2^{(m-1)/2}$    & $ 2^{(m-7)/2} (4 - 2^{2 + m} + 2^{(1 + 3 m)/2})$ \\ [2mm]
\hline
\end{tabular}
\end{center}
\end{table}

\begin{table}[ht]
\begin{center}
\caption{The weight distribution of $\C_T$ for $m$ odd and $t =3$}\label{tab-des13}
\begin{tabular}{cc} \hline
Weight  &  Multiplicity   \\ \hline
$0$          &  $1$ \\  [2mm]
$2^{m-1}-2^{(m-1)/2}$    & $ -2^{( m-3)/2} + 3\cdot 2^{(3m-7)/2} + 2^{m-3} + 2^{2m-4}$  \\ [2mm]
$2^{m-1}$    & $(-1 + 2^{m-2}) (1+2^{ m-1})$ \\ [2mm]
$2^{m-1}+2^{(m-1)/2}$    & $ 2^{( m-3)/2} - 3\cdot 2^{(3m-7)/2} + 2^{m-3} + 2^{2m-4}$  \\ [2mm]
\hline
\end{tabular}
\end{center}
\end{table}

\begin{example}\label{exa-des11}
Let $m=5$ and $T$ be a $1$-subset of $\mathcal P(\C)$. Then the shortened code $\C_{T}$  in Theorem \ref{main-design1} is a $[31,10,12]$ binary linear code with the weight enumerator $1+310z^{12}+527z^{16}+186z^{20}$. The code $\C_{T}$ is optimal. The dual code of $\C_{T}$ has parameters $[31,21,5]$ and is optimal according to the tables of best known codes maintained at http://www.codetables.de.
\end{example}

\begin{example}\label{exa-des12}
Let $m=5$ and $T$ be a $2$-subset of $\mathcal P(\C)$. Then the shortened code $\C_{T}$  in Theorem \ref{main-design1} is a $[30,9,12]$ linear code with the weight enumerator $1+190z^{12}+255z^{16}+66z^{20}$. The code $\C_{T}$ is optimal. The dual code of $\C_{T}$ has parameters $[30,21,4]$ and is optimal according to the tables of best known codes   maintained at http://www.codetables.de.
\end{example}

\begin{example}\label{exa-des13}
Let $m=5$ and $T$ be a $3$-subset of $\mathcal P(\C)$. Then the shortened code $\C_{T}$  in Theorem \ref{main-design1} is a $[29,8,12]$ binary linear code with the weight enumerator $1+114z^{12}+119z^{16}+22z^{20}$. The code $\C_{T}$ is optimal. The dual code of $\C_{T}$ has parameters $[29,21,3]$ and is almost optimal according to the tables of best known codes   maintained at http://www.codetables.de.
\end{example}

\begin{theorem}\label{main-design2}
Let $m\geq 4 $ be even, and $\C$ be a binary linear code with length $2^m$  and the weight distribution in Table \ref{tab-cf1}. Let  $T$ be a $t$-subset of $\mathcal P(\C)$.
We have the following results.
\begin{itemize}
  \item   If $t=1$, then the shortened code $\C_{T}$ is a $[2^{m}-1, 2m,   2^{m-1} - 2^{m/2}]$ binary linear code with the weight distribution in Table \ref{tab-des21}.
  \item   If $t=2$, then the shortened code $\C_{T}$ is a $[2^{m}-2, 2m-1, 2^{m-1} - 2^{m/2}]$ binary linear code with the weight distribution in Table \ref{tab-des22}.
\end{itemize}
\end{theorem}

\begin{table}[ht]
\begin{center}
\caption{The weight distribution of $\C_T$ for $m$  even and $t=1$ }\label{tab-des21}
\begin{tabular}{cc} \hline
Weight  &  Multiplicity   \\ \hline
$0$          &  $1$ \\ [2mm]
$2^{m-1}-2^{m/2}$    & $ 1/3 \cdot 2^{-3 + m/2} (2 + 2^{m/2}) (-1 + 2^m) $ \\ [2mm]
$2^{m-1}-2^{(m-2)/2}$    & $ 1/3 \cdot 2^{m/2} (-1 + 2^{m/2}) (1 + 2^{m/2})^2 $ \\ [2mm]
$2^{m-1}$    & $ (2^m-1) (1+2^{ m-2})$ \\            [2mm]
$2^{m-1}+2^{(m-2)/2}$    & $ 1/3 \cdot 2^{m/2} (-1 + 2^{m/2})^2  (1 + 2^{m/2})   $ \\ [2mm]
$2^{m-1}+2^{m/2}$    &   $ 1/3 \cdot 2^{-3 + m/2} (-2 + 2^{m/2}) (-1 + 2^m)  $ \\ [2mm]
\hline
\end{tabular}
\end{center}
\end{table}

\begin{table}[ht]
\begin{center}
\caption{The weight distribution of $\C_T$ for $t=2$ and $m$ even }\label{tab-des22}
\begin{tabular}{cc} \hline
Weight  &  Multiplicity   \\ \hline
$0$          &  $1$ \\ [2mm]
$2^{m-1}-2^{m/2}$    & $ 1/3 \cdot 2^{m/2-4} (2 + 2^{m/2}) (2^m +2^{1+m/2}-2) $ \\ [2mm]
$2^{m-1}-2^{(m-2)/2}$    & $ 1/3 \cdot 2^{m/2-1} (1 + 2^{m/2}) (2^m +2^{m/2}-2)  $ \\ [2mm]
$2^{m-1}$                 & $ (2^{m-1}-1)(1+2^{ m-2})$ \\            [2mm]
$2^{m-1}+2^{(m-2)/2}$    & $ 1/3 \cdot 2^{m/2-1} (-1 + 2^{m/2}) (2^m -2^{m/2}-2)  $ \\ [2mm]
$2^{m-1}+2^{m/2}$    & $ 1/3 \cdot 2^{m/2-4} (4 + 2^{1 + m/2} + 2^{3 m/2} - 2^{2 + m})$ \\ [2mm]
\hline
\end{tabular}
\end{center}
\end{table}

\begin{example}\label{exa-des21}
Let $m=4$ and $T$ be a $1$-subset of $\mathcal P(\C)$. Then the shortened code $\C_{T}$  in Theorem \ref{main-design2} is a $[15,8,4]$ linear code with the weight enumerator $1+15z^{4}+100z^{6}+75 z^{8}+60 z^{10}+5 z^{12}$. This code $\C_{T}$ is optimal. Its dual $\C_{T}^\perp$ has parameters $[15,7,5]$ and is optimal according to the tables of best known codes   maintained at http://www.codetables.de.
\end{example}


\begin{example}\label{exa-des22}
Let $m=4$ and $T$ be a $2$-subset of $\mathcal P(\C)$. Then the shortened code $\C_{T}$  in Theorem \ref{main-design2} is a $[14,7,4]$ binary linear code with the weight enumerator $1+11z^{4}+60z^{6}+35z^{8}+20 z^{10}+z^{12}$. This code $\C_{T}$ is optimal. Its dual $\C_{T}^\perp$ has parameters $[14,7,4]$ and is optimal according to the tables of best known codes   maintained at http://www.codetables.de.
\end{example}


\subsection{Several general results on shortened codes }

\begin{lemma}\label{lem-cf}
Let $m\geq 5$ be odd (resp., $m\geq 4$ be even), and $\C$ be a binary linear code with the length $2^m$  and the weight distribution in Table \ref{tab-cf} (resp., Table \ref{tab-cf1}). Then the dual code $\C^\bot$ of $\C$ has parameters $[2^m,2^m-2m-1,6]$.
\end{lemma}
\begin{proof}
The weight distribution in Table \ref{tab-cf} (or \ref{tab-cf1}) means that the dimension of $\C$ is $2m+1$. Thus, the dual code $\C^\bot$ of $\C$ has dimension $2^m-2m-1$. Since the code length of $\C$ is $2^m$, from the weight distribution in Table \ref{tab-cf} (or \ref{tab-cf1}) and the first seven Pless power moments in (\ref{eq:PPM}), it is easily obtain that $A_6(\C^\perp)> 0$ and  $A_i(\C^\perp)=0$ for any $i\in \{1,2,3,4,5\}$.
The desired conclusions then follow .
\end{proof}


\begin{theorem}\label{main-31}
Let $m\geq 4$, and $\C$ be a binary linear code with length $2^m$  and the weight distribution in Table \ref{tab-cf} for odd $m$
and Table \ref{tab-cf1} for even $m$. Let  $T$ be a $4$-subset of $\mathcal P(\C)$ and  $\lambda _{T,6}(\C^\bot)=\lambda$,
then $\lambda =0$ or $1$. Furthermore, we have the following results.
\begin{itemize}
  \item [(\uppercase\expandafter{\romannumeral1})]  If $m\geq 5$ is odd and $\lambda =0$, then the shortened code $\C_{T}$ is a $[2^{m}-4, 2m-3, 2^{m-1} - 2^{(m -1)/2}]$ binary linear code with the weight distribution in Table \ref{tab-31}.
  \item [(\uppercase\expandafter{\romannumeral2})]  If  $m\geq 5$ is odd and $\lambda =1$, then the shortened code $\C_{T}$ is a $[2^{m}-4, 2m-3, 2^{m-1} - 2^{(m -1)/2}]$ binary linear code with the weight distribution in Table \ref{tab-32}.
\end{itemize}
\end{theorem}

\begin{table}[ht]
\begin{center}
\caption{The weight distribution of $\C_T$ for $\lambda =0$ }\label{tab-31}
\begin{tabular}{cc} \hline
Weight  &  Multiplicity   \\ \hline
$0$          &  $1$ \\  [2mm]
$2^{m-1}-2^{(m-1)/2}$    & $ -2^{( m-3)/2} + 2^{m-3} + 2^{2 m-5} + 2^{(3 m-5)/2}$  \\ [2mm]
$2^{m-1}$    & $-1 - 2^{m-2} + 4^{ m-2}$ \\ [2mm]
$2^{m-1}+2^{(m-1)/2}$    & $ 2^{( m-3)/2} + 2^{m-3} + 2^{2 m-5} - 2^{(3 m-5)/2}$ \\ [2mm]
\hline
\end{tabular}
\end{center}
\end{table}

\begin{table}[ht]
\begin{center}
\caption{The weight distribution of $\C_T$ for $\lambda =1$ }\label{tab-32}
\begin{tabular}{cc} \hline
Weight  &  Multiplicity   \\ \hline
$0$          &  $1$ \\ [2mm]
$2^{m-1}-2^{(m-1)/2}$    & $ 3\times 2^{m-4} - 2^{(-3 + m)/2} + 2^{-5 + 2 m} + 2^{(-5 + 3 m)/2}$ \\ [2mm]
$2^{m-1}$    & $ 2^{-4} (-8 + 2^m) (2 + 2^m)$ \\ [2mm]
$2^{m-1}+2^{(m-1)/2}$    & $ 3\times 2^{m-4} + 2^{(-3 + m)/2} + 2^{-5 + 2 m} - 2^{(-5 + 3 m)/2}$ \\ [2mm]
\hline
\end{tabular}
\end{center}
\end{table}

\begin{proof}
By the definition of $\Lambda _{T,6}(\C^\perp)$, we have
$$\lambda=\lambda _{T,6}(\C^\perp )=\# \left \{\mathrm{Supp}(\mathbf c): ~\mathbf{c}\in \C^\perp,~wt(\mathbf{c})=6 ~and~T\subseteq \mathrm{Supp}(\mathbf c) \right \}. $$
If  $\lambda \geq2$, there would be  $\mathrm{Supp}(\mathbf c_1), \mathrm{Supp}(\mathbf c_2)\in \Lambda _{T,6}(\C^\perp)$.  Then $\mathbf c_1+\mathbf c_2 \in \C^\bot$ and the weight $wt(\mathbf c_1+\mathbf c_2)\leq 4$. This is a contradiction to the minimum distance $6$ of  $\C^\bot$ in Lemma \ref{lem-cf}. Thus, $\lambda =0$ or $1$.

We treat the weight distribution of $\C_T$ according to the value of $\lambda$ as follows.

(\uppercase\expandafter{\romannumeral1}) The case that $\lambda=0$ and $m$ is odd.

By Lemma \ref{lem-cf}, the minimum distance of  $\C^\bot$ is 6. Thus,
\begin{align}\label{eq-A}
A_1\left ( \left (\mathcal C^{\perp} \right )^{T}  \right )=A_2\left ( \left (\mathcal C^{\perp} \right )^{T}  \right )=0,~
A_1\left ( \left (\mathcal C_{T} \right )^{\perp}  \right )=A_2\left ( \left (\C_{T} \right )^{\perp}  \right )=0
\end{align}
and the shortened code $\C_{T}$ has length $n=2^m-4$ and dimension $k=2m-3$ from $\lambda _{T,6}(\C^\perp )=0$  and Lemma \ref{lem:C-S-P}. By definition and Lemma \ref{lem-cf}, we have
$A_i\left (\C_T \right )=0$ for $i\not \in \{0, i_1, i_2, i_3\} $, where $i_1=2^{m-1}-2^{(m-1)/2}$, $i_2=2^{m-1}$
and $i_3=2^{m-1}+2^{(m-1)/2}$. Therefore, from (\ref{eq-A}) and (\ref{eq:PPM}),
the first three Pless power moments
\begin{align*}
\left\{
  \begin{array}{l}
    A_{i_1} + A_{i_2} +A_{i_3} = 2^{2m-3}-1,  \\
    i_1 A_{i_1} + i_2  A_{i_2} + i_3 A_{i_3}  = 2^{2m-3-1}(2^m-4), \\
    i_1^2 A_{i_1} + i_2^2  A_{i_2} + i_3^2 A_{i_3}  = 2^{2m-3-2}(2^m-4)(2^m-4+1).
  \end{array}
\right.
\end{align*}
yield the weight distribution in Table \ref{tab-31}. This completes the proof of (\uppercase\expandafter{\romannumeral1}).

(\uppercase\expandafter{\romannumeral2}) The case that $\lambda =1$ and $m$ is odd.

The proof is similar to that of (\uppercase\expandafter{\romannumeral1}). Since $\lambda _{T,6}(\C^\perp )=1$  and the minimum  distance of  $\C^\bot$ is 6, from Lemma \ref{lem:C-S-P} we have
\begin{align}\label{eq-AA}
& A_1\left ( \left (\mathcal C^{\perp} \right )^{T}  \right )=0,~A_2\left ( \left (\mathcal C^{\perp} \right )^{T}  \right )=1,~   \nonumber  \\
& A_1\left ( \left (\mathcal C_{T} \right )^{\perp}  \right )=0,~~A_2\left ( \left (\C_{T} \right )^{\perp}  \right )=1.
\end{align}
Then the desired conclusions follow from (\ref{eq-AA}), the definitions and the first three Pless power moments of (\ref{eq:PPM}). This completes the proof.
\end{proof}

\begin{lemma}\label{le-even31}
Let $m\geq 4$ be even, and $\C$ be a binary linear code with length $2^m$  and the weight distribution in Table \ref{tab-cf1}. Let  $T$ be a $3$-subset of $\mathcal P(\C)$. Suppose $\lambda _{T,6}(\C^\bot)=\lambda$, then
$ A_1\left ( \left (\mathcal C^{\perp} \right )^{T}  \right )= A_2\left ( \left (\mathcal C^{\perp} \right )^{T}  \right )=0$, $A_3\left ( \left (\mathcal C^{\perp} \right )^{T}  \right )=\lambda$ and
$ A_4\left ( \left (\mathcal C^{\perp} \right )^{T}  \right )=2\cdot (2^{m-2}-1)^2-  3\lambda $.
\end{lemma}

\begin{proof}
By Lemma \ref{lem-cf}, the minimum distance of  $\C^\bot$ is 6. Thus, from $\#T=3$ and the definition of $\lambda _{T,6}(\C^\bot)$, we have
$
A_1\left ( \left (\mathcal C^{\perp} \right )^{T}  \right )=A_2\left ( \left (\mathcal C^{\perp} \right )^{T}  \right )=0
$
and
$A_3\left ( \left (\mathcal C^{\perp} \right )^{T}  \right )=\lambda$. Note that the code $\C$ has length $2^m$ and dimension $2m+1$. By Lemma \ref{lem-cf}, Table \ref{tab-cf1} and the first seven Pless power moments of (\ref{eq:PPM}), we have
$$
A_6(\C^{\perp})= \frac{1}{45}\cdot  2^{m-4} (2^m-4)^2 (2^m-1).
$$
Further, from Theorem \ref {thm-AMTheorem}, Lemmas \ref{lem:P:k:k+t} and \ref{lem-cf}, we conclude deduce that
 $(\mathcal P(\mathcal C^{\perp}), \mathcal B_{6}(\mathcal C^{\perp}))$ is a $2$-design and
$$ A_4\left ( \left (\mathcal C^{\perp} \right )^{\{t_1,t_2\}}  \right )=\frac{\binom{  6}{ 2 }}{\binom{ q }{ 2 }}\cdot A_6(\C^{\perp}) =\frac{2}{3}\cdot (2^{m-2}-1)^2$$
for any $\{t_1,t_2\} \subseteq \mathcal P(\C)$. Let $T=\{t_1,t_2,t_3\}$. Since $\#T=3$ and the minimum distance of  $\C^\bot$ is $6$,
an easy computation shows that
\begin{eqnarray*}
\begin{array}{rl}
 A_4\left ( \left (\mathcal C^{\perp} \right )^{T}  \right )&=\sum_{1\le i < j \le 3}
  \left ( A_4\left ( \left (\mathcal C^{\perp} \right )^{\{t_i,t_j\}}  \right )-  \lambda _{T,6}(\C^\bot) \right )\\
 &= \binom{ 3}{ 2 }\left ( A_4\left ( \left (\mathcal C^{\perp} \right )^{\{t_1,t_2\}}  \right )-  \lambda \right ).
 \end{array}
 \end{eqnarray*}
Then the desired conclusions follow.
\end{proof}

\begin{theorem}\label{main-even32}
 Let $m\geq 4$ be even, and $\C$ be a binary linear code with length $2^m$  and the weight distribution in Table \ref{tab-cf1}. Let  $T$ be a $3$-subset of $\mathcal P(\C)$. Suppose $\lambda _{T,6}(\C^\bot)=\lambda$, then the shortened code $\C_{T}$ is a $[2^{m}-3, 2m-2, 2^{m-1}-2^{m/2}]$ binary linear code with the weight distribution in Table \ref{tab-even32}.
\end{theorem}

\begin{proof}
The proof is similar to that of Theorem \ref{main-31}. From Lemma \ref{lem:C-S-P} and $\#T=3$, the shortened code $\C_{T}$ has length $n=2^m-3$ and dimension $k=2m-2$.
By definition and the weight distribution in Table \ref{tab-cf1}, we have
$A_i\left (\C_T \right )=0$ for $i\not \in \{0, i_1, i_2, i_3,i_4,i_5\} $, where $i_1=2^{m-1}-2^{m/2}$, $i_2=2^{m-1}-2^{(m-2)/2}$
, $i_3=2^{m-1}$, $i_4=2^{m-1}+2^{(m-2)/2}$  and $i_5=2^{m-1}+2^{m/2}$.
Moreover, from Lemmas \ref{lem:C-S-P} and \ref{le-even31} we have
$ A_1\left ( \left (\mathcal C_{T} \right )^{\perp}  \right )= A_2\left ( \left (\mathcal C_{T} \right )^{\perp}  \right )=0$, $A_3\left ( \left (\mathcal C_{T} \right )^{\perp} \right )=\lambda$
and
$ A_4\left ( \left (\mathcal C_{T} \right )^{\perp} \right )=2\cdot (2^{m-2}-1)^2-  3\lambda $.
Therefore, the first five Pless power moments of (\ref{eq:PPM}) yield the weight distribution in Table \ref{tab-even32}. This completes the proof.
\end{proof}

\begin{table}[ht]
\begin{center}
\caption{The weight distribution of $\C_T$ for $\lambda _{T,6}(\C^\bot)=\lambda$  }\label{tab-even32}
\begin{tabular}{cc} \hline
Weight  &  Multiplicity   \\ \hline
$0$          &  $1$ \\ [2mm]
$2^{m-1}-2^{m/2}$    & $ 1/3 \cdot 2^{ m/2-5} (8 + 2^{3 + m/2} + 2^{3 m/2} + 2^{2 + m} + 12 \lambda)$ \\ [2mm]
$2^{m-1}-2^{(m-2)/2}$    & $ 1/3 \cdot 2^{m/2-3} ((2 + 2^{m/2}) (-8 + 3 \cdot 2^{m/2} + 2^{1 + m}) - 6 \lambda)$ \\ [2mm]
$2^{m-1}$    & $ -1 + 4^{ m-2}$ \\            [2mm]
$2^{m-1}+2^{(m-2)/2}$    & $ 1/3 \cdot 2^{m/2-3} ((-2 + 2^{m/2}) (-8 - 3 \cdot 2^{m/2} + 2^{1 + m}) + 6 \lambda) $ \\ [2mm]
$2^{m-1}+2^{m/2}$    & $ 1/3 \cdot 2^{ m/2-5} (-8 + 2^{3 + m/2} + 2^{3 m/2} - 2^{2 + m} - 12 \lambda)$ \\ [2mm]
\hline
\end{tabular}
\end{center}
\end{table}

\section{Shortened linear codes from APN functions} \label{sec-APN}

Let $p=2$ and $q=2^m$. In this section, we study some shortened codes $\C_{T}$ of linear codes $\C$ defined by (\ref{eq:cf}) and determine their parameters
for the case that $f(x)$ is an APN monomial function $x^{2^e+1}$. It is known that $\C$ has the weight distribution in Tables \ref{tab-cf} (resp. Tables \ref{tab-cf1}) when $m$ is odd (resp. $m$ is even).

Let $T$ be a $t$-subset of $\cP(\mathcal C):=\gf(q)$. We will consider some shortened codes $\C_{T}$ of $\C$ for the cases $t=3~ \mbox{or} ~4$.

\subsection{Some shortened codes for the case $t=4$ and $m$  odd}

We notice that it is difficult to determine the value of  $ \lambda _{T,6}(\C^\perp ) $ in Theorem \ref{main-31} for general APN function $f$. We will
determine $ \lambda _{T,6}(\C^\perp ) $  for  APN function $f(x)=x^{2^e+1}$. To this end, the following lemma will be needed.

\begin{lemma}\label{lem-solu1}
Let $e$  and $m\geq 4$ be positive integers with $\gcd(m,e)=1$. Let $q=2^m$ and  $ \{ x_1, x_2, x_3, x_4\}$ a $4$-subset of  $ \gf(q) $. Denote $S_i=x_1^i+x_2^i+x_3^i+x_4^i$.
Let $N$ be the number of solutions $(x,y)\in \gf(q)^2$ of the system of equations
\begin{align}\label{eqx3}
\left\{
  \begin{array}{cccl}
    x&+&y&=S_1,  \\ [2mm]
    x^{2^e+1} &+&y^{2^e+1}&=S_{2^e+1}, \\[2mm]
   \multicolumn{3}{c}{ \# \{x_1, x_2, x_3, x_4, x,y\}} & =6.
  \end{array}
\right.
\end{align}
Then $N =2$ if  $S_1 \neq 0$ and $\tr_{q/2}\left (\frac{S_{2^e+1}}{S_1^{2^e+1}}+1 \right )=0$ , and $N=0$ otherwise.
\end{lemma}

\begin{proof}
Let us denote by $\C$  the linear code from $f(x)=x^{2^e+1}$ given in (\ref{eq:cf}).
By Lemma \ref{lem-cf}, the minimum weight of the dual code $\C^{\perp}$ is equal to $6$.
Consequently, (\ref{eqx3}) is equivalent to the following system of equation
\begin{align}\label{eqx3-2}
\left\{
  \begin{array}{cccl}
    x&+&y&=S_1,  \\ [2mm]
    x^{2^e+1} &+&y^{2^e+1}&=S_{2^e+1}.
  \end{array}
\right.
\end{align}
Substituting $y= x + S_1$ into the second equation of (\ref{eqx3-2}) leads to
\begin{align}\label{eq-con2}
    & x^{2^e+1}+(x+S_1)^{2^e+1}+S_{2^e+1}  \nonumber \\
    & = x^{2^e+1}+(x^{2^e}+S_1^{2^e})(x+S_1)+S_{2^e+1}  \nonumber \\
    &=S_1 x ^{2^e}+S_1^{2^e}  x+ S_1^{2^e+1}+S_{2^e+1}\nonumber \\
    &=0.
\end{align}
We claim that $S_1\neq 0$ if $N\neq 0$. On the contrary, suppose that $N\neq 0$
and $S_1=0$. Now (\ref{eq-con2}) clearly forces $S_{2^e+1}=0$. It follows that the four coordinate positions $x_1, x_2, x_3, x_4$ give rise to a
codeword of weight $4$ of $\C^{\perp}$. This contradicts the fact that the minimum weight of
$\C^{\perp}$ equals $6$. Therefore $S_1\neq 0$ if $N\neq 0$. In particular, $N=0$ if $S_1=0$.

When $S_1\neq 0$, Equation (\ref{eq-con2})   is equivalent to
\begin{align}\label{eq-con3}
     \frac{S_1 x ^{2^e}+S_1^{2^e} \cdot x+ S_1^{2^e+1}+S_{2^e+1}}{S_1^{2^e+1}}=\left (\frac{x}{S_1} \right )^{2^e}+\frac{x}{S_1}+1+\frac{S_{2^e+1}}{S_1^{2^e+1}}=0.
\end{align}

If $S_1 \neq 0$ and $\tr_{q/2}\left (\frac{S_{2^e+1}}{S_1^{2^e+1}}+1 \right )\neq 0$, it may be concluded that
there is no solution in $\gf(q)$ to Equation (\ref{eq-con3}). Thus $N=0$.

If $S_1 \neq 0$ and $\tr_{q/2}\left (\frac{S_{2^e+1}}{S_1^{2^e+1}}+1 \right )= 0$, we see that Equation (\ref{eq-con3})
 has exactly two different solutions $x, x+S_1\in \gf(q)$  from $\gcd (m,e)=1$.
This means that Equation (\ref{eqx3-2})  has exactly two different solutions $(x, x+S_1)$ and $(x+S_1,x)$ in $\gf(q)^2$. Therefore $N=2$.
This completes the proof.
\end{proof}


\begin{lemma}\label{main-32}
Let $e$ and $m\geq 4$ be positive integers with $\gcd(m,e)=1$. Let $q=2^m$,  $f(x)=x^{2^e+1}$ and $\C$ be defined in (\ref{eq:cf}). Let $T=\{x_1,x_2,x_3,x_4\}$ be a $4$-subset of $\mathcal P(\C)$ . Then $\lambda _{T,6}(\C^\bot)=1$ if $\sum_{i=1}^4 x_i \neq 0$ and
$\tr_{q/2} \left(\sum_{i=1}^4 x_i^{2^e+1}/(\sum_{i=1}^4 x_i)^{2^e+1}+1 \right)=0$,  and $\lambda _{T,6}(\C^\bot)=0$ otherwise.
\end{lemma}

\begin{proof}
By definition, the dual code $\C^\bot$ of $\C$ has minimum distance 6.
By definition we have $\lambda _{T,6}(\C^\bot)=\frac{N}{2!}$, where $N$ was defined in Lemma \ref{lem-solu1}. Then the desired conclusions follow from Lemma \ref{lem-solu1}.
\end{proof}

By Theorem \ref{main-31} and Lemma \ref{main-32},
we have the following theorem, which is one of the main results in this paper.
\begin{theorem}\label{main-odd4}
Let $m\geq 5$ be odd and  $e$  be a positive integer with $\gcd(m,e)=1$. Let $q=2^m$,  $f(x)=x^{2^e+1}$ and $\C$ be defined in (\ref{eq:cf}).
Let $T=\{x_1,x_2,x_3,x_4\}$ be a $4$-subset of $\mathcal P(\C)$ .
Then $\C_T$ has the weight distribution of Table \ref{tab-31} if $\sum_{i=1}^4 x_i \neq 0$ and $\tr_{q/2}\left(\sum_{i=1}^4 x_i^{2^e+1}/(\sum_{i=1}^4 x_i)^{2^e+1}\right)=1$,
and  Table \ref{tab-32} otherwise.
\end{theorem}

\begin{example}\label{exa-31odd}
Let $m=5$, $q=2^5$ and $\alpha$ be a primitive element of $\gf(q)$ with minimum polynomial $\alpha^5+ \alpha^2+1=0$. Let $e=1$ and $T=\{\alpha^1,\alpha^2,\alpha^4,\alpha^5\}$. Then $\gamma =\alpha^{17}$,
$\tr_{q/2} \left(\frac{\bar{\gamma}}{\gamma^3} \right)=0$ and $\lambda _{T,6}(\C^\bot)=0$, where
$\gamma = \alpha^1+\alpha^2+\alpha^4+\alpha^5$ and $\bar{\gamma} = \alpha^3+\alpha^6+\alpha^{12}+\alpha^{15}$. The shortened code $\C_{T}$  in Theorem \ref{main-odd4} is a $[28,7,12]$ binary linear code with the weight enumerator $1+66z^{12}+55z^{16}+6z^{20}$. The code $\C_{T}$ is optimal according to the tables of best known codes   maintained at http://www.codetables.de.
\end{example}

\begin{example}\label{exa-32odd}
Let $m=5$, $q=2^5$ and $\alpha$ be a primitive element of $\gf(q)$ with minimal polynomial $\alpha^5+ \alpha^2+1=0$. Let $e=1$ and $T=\{\alpha^1,\alpha^2,\alpha^3,\alpha^4\}$. Then $ \gamma=\alpha^{24}$,
$\tr_{q/2} \left(\frac{\bar{\gamma}}{\gamma^3} \right)=1$, and $\lambda _{T,6}(\C^\bot)=1$, where $\gamma = \alpha^1+\alpha^2+\alpha^3+\alpha^4$ and $\bar{\gamma} = \alpha^3+\alpha^6+\alpha^9+\alpha^{12}$. The shortened code $\C_{T}$  in Theorem \ref{main-odd4} is a $[28,7,12]$ binary linear code with the weight enumerator $1+68z^{12}+51z^{16}+8z^{20}$.
The code $\C_{T}$ is optimal according to the tables of best known codes   maintained at http://www.codetables.de.
\end{example}

\subsection{Some shortened codes for the case $t=3$ and $m$ even}

For any $T=\{x_1,x_2,x_3\}\subseteq \mathcal P(\C)$,  we notice that it is difficult to determine the value of $\lambda$ in Theorem \ref{main-even32}. We will
study a class of special linear codes $\C$ from the APN monomial functions $x^{2^e+1}$   defined by (\ref{eq:cf}).
In the following, we will determine the value of $\lambda$ and the parameters of the shortened code $\C_{T}$. We need the result in the following lemma.

\begin{lemma}\label{lem-solu2}
Let $e$ be a positive integer, $m$ be even with $\gcd(m,e)=1$, and  $q=2^m$. Let
$\hat{N}$ be the number of solutions $(x_1, x_2, x_3)\in \gf(q)^3$ of the system of equations
\begin{align}\label{eqx4}
\left\{
  \begin{array}{cccccl}
    x_1&+&x_2&+&x_3&=a,  \\ [2mm]
    x_1^{2^e+1}&+&x_2^{2^e+1}&+&x_3^{2^e+1}&=b, \\
  \end{array}
\right.
\end{align}
where $a,b\in \gf(q)$ and $a^{2^e+1}\neq b$. Then $\hat{N} =2^m+(-2)^{m/2}-2$ if $a^{2^e+1}+b$  is not a cubic residue, and $\hat{N}=2^m+(-2)^{m/2+1}-2$ if $a^{2^e+1}+b$  is  a cubic residue.
\end{lemma}

\begin{proof}
Replacing $x_1$ with $x + a$, $x_2$ with $y + a$ and $x_3$ with $z + a$, we have
\begin{align}\label{eqx5}
\left\{
  \begin{array}{cccccrc}
    x&+&y&+&z&=&0,  \\ [2mm]
    x^{2^e+1}&+&y^{2^e+1}&+&z ^{2^e+1}&=&a^{2^e+1}+b.
  \end{array}
\right.
\end{align}
Substituting $z= x + y$ into the second equation of (\ref{eqx5}) yields to
\begin{align}\label{eqx6}
x^{2^e}y+y^{2^e}x=a^{2^e+1}+b.
\end{align}
Thus, $\hat{N}$ equals the number of solutions $(x,y)\in \gf(q)^2$ to Equation (\ref {eqx6}).

Since $a^{2^e+1}+b \neq 0$, replacing $y$ with $xy'$ in Equation (\ref {eqx6}) gives
\begin{align}\label{eqx7}
x^{2^e}y'+y'^{2^e}x= x^{2^e+1}y'+y'^{2^e}x^{2^e+1}=x^{2^e+1}(y'+y'^{2^e}) =a^{2^e+1}+b.
\end{align}
A rearrangement of Equation (\ref {eqx7}) yields
\begin{align}\label{eqn:replacing}
y'+y'^{2^e} =(a^{2^e+1}+b)x^{-(2^e+1)}.
\end{align}
Then
\begin{align}\label{eqx8}
\tr \left ((a^{2^e+1}+b)x^{-(2^e+1)} \right )=0.
\end{align}
Further, from Lemma \ref{expsum} we get
\begin{align}\label{eqx9}
&\# \left \{x\in \gf(q)^*: \tr \left ((a^{2^e+1}+b)x^{-(2^e+1)} \right )=0 \right \}   \nonumber  \\
&= \# \left \{x\in \gf(q)^*: \tr \left ((a^{2^e+1}+b)x^{(2^e+1)} \right )=0 \right \}   \nonumber  \\
&=\frac{1}{2}\cdot \sum_{u\in \gf(2)}\sum_{x\in \gf(q)^*}(-1)^{\tr \left (u(a^{2^e+1}+b)x^{(2^e+1)} \right )}    \nonumber   \\
&=\frac{1}{2} \left (q-2+\sum_{x\in \gf(q)}(-1)^{\tr \left ((a^{2^e+1}+b)x^{(2^e+1)} \right )} \right )   \nonumber  \\
&=\left\{
\begin{array}{ll}
\frac{1}{2}\left (q-2+(-1)^{m/2} 2^{m/2} \right ),  &\mbox{ if $a^{2^e+1}+b$ is not a cubic residue,} \\ [2mm]
\frac{1}{2}\left (q-2-(-1)^{m/2} 2^{m/2+1} \right ), & \mbox{ if $a^{2^e+1}+b$ is a cubic residue.}
\end{array}
\right.
\end{align}
Since $\gcd(m,e)=1$, it follows easily that if $\tr \left ((a^{2^e+1}+b)x^{-(2^e+1)} \right )=0$, where $x\in \gf(q)^*$, then there exactly exist two $y'$ in $\gf(q)$
satisfying $y'+y'^{2^e} =(a^{2^e+1}+b)x^{-(2^e+1)}$. Then the desired conclusions follow from Equation (\ref{eqx9}).
\end{proof}

\begin{theorem}\label{main-even33}
Let $e$ be a positive integer and $m\geq 4$ be even with $\gcd(m,e)=1$.
Let $q=2^m$,  $f(x)=x^{2^e+1}$ and $\C$ be defined in (\ref{eq:cf}). Let $T=\{x_1,x_2,x_3\}$ be a $3$-subset of $\mathcal P(\C)$. Then
\begin{align}
\lambda=\lambda _{T,6}(\C^\bot) &=\left\{
\begin{array}{ll}
\frac{1}{6}\left (q-2+(-1)^{m/2} 2^{m/2} \right )-1, &    \mbox{ if $a^{2^e+1}+b$ is not a cubic residue} \\[2mm]
\frac{1}{6}\left (q-2-(-1)^{m/2} 2^{m/2+1} \right )-1,    &\mbox{ if $a^{2^e+1}+b$ is a cubic residue,}
\end{array}
\right.
\end{align}
where $a=\sum_{i=1}^3 x_i$ and $b=\sum_{i=1}^3 x_i^{2^e+1}$. Moreover, the shortened code $\C_{T}$ is a $[2^{m}-3, 2m-2, 2^{m-1}-2^{m/2}]$ binary linear code with the weight distribution in Table \ref{tab-even32}.
\end{theorem}

\begin{proof}
By definition,  the code $\C$ has parameters $[2^m,2m+1,2^{m-1}-2^{m/2}]$ and the weight distribution in Table \ref{tab-cf1}. By Lemma \ref{lem-cf},
the minimum distance of  $\C^\bot$ is 6.
Consider the system of equations given by
\begin{align}\label{eqxexp}
\left\{
  \begin{array}{cccccl}
    x&+&y&+&z&=a,  \\ [2mm]
    x^{2^e+1}&+&y^{2^e+1}&+&z^{2^e+1}&=b. \\
  \end{array}
\right.
\end{align}
Since $(x_1,x_2,x_3)$ must be the solution $(x,y,z)\in \gf(q)^3$ of Equation (\ref{eqxexp}),
by definition we have $\lambda _{T,6}(\C^\bot)=\frac{\hat{N}}{3!}-1$, where $\hat{N}$ was defined in Lemma \ref{lem-solu2}. Then the desired conclusions follow from Lemma \ref{lem-solu2} and Theorem \ref{main-even32}.
\end{proof}

\begin{example}\label{exa-31even}
Let $m=4$, $q=2^4$ and $\alpha$ be a primitive element of $\gf(q)$ with minimal polynomial $\alpha^4+ \alpha+1=0$. Let $e=1$ and $T=\{\alpha^1,\alpha^2,\alpha^4\}$.
Then $(\alpha^1+\alpha^2+\alpha^4)^3+(\alpha^3+\alpha^6+\alpha^{12})=1$, $\lambda=0$ and the shortened code $\C_{T}$  in Theorem \ref{main-even32} is a $[13,6,4]$ binary linear code with the weight enumerator $1+7z^{4}+36z^{6}+15z^{8}+4 z^{10}+z^{12}$. This code $\C_{T}$ is optimal. Its dual $\C_{T}^\perp$ has parameters $[13,7,4]$ and is optimal according to the tables of best known codes   maintained at http://www.codetables.de.
\end{example}

\begin{example}\label{exa-32even}
Let $m=4$, $q=2^4$ and $\alpha$ be a primitive element of $\gf(q)$ with minimal polynomial $\alpha^4+ \alpha+1=0$. Let $e=1$ and $T=\{\alpha^2,\alpha^5,\alpha^7\}$.
Then $(\alpha^2+\alpha^5+\alpha^7)^3+(\alpha^6+\alpha^{15}+\alpha^{21})=\alpha^{11}$, $\lambda=2$ and the shortened code $\C_{T}$  in Theorem \ref{main-even32} is a $[13,6,4]$ binary linear code with the weight enumerator $1+8z^{4}+34z^{6}+15z^{8}+6 z^{10}$.
This code $\C_{T}$ is optimal. Its dual $\C_{T}^\perp$ has parameters $[13,7,3]$ and is almost optimal according to the tables of best known codes   maintained at http://www.codetables.de.
\end{example}

%

\subsection{Some shortened codes for the case $t=4$ and $m$   even}

Let $T=\{x_1,x_2,x_3,x_4\}$ be a $4$-subset of $ \mathcal P(\C)$.  Magma programs show that the weight distributions of $\C_{T}$ for the codes $\C$ from
APN functions are very complex.
Thus, it is difficult to determine their parameters in general.
In this subsection, we will study a class of special linear codes $\C$ with the weight distribution in Table \ref{tab-cf1} and
determine the parameters of $\C_{T}$ for certain $4$-subsets $T$ in Theorem \ref{main-even41}.

In order to determine the parameters of $\C_{T}$, we need the next two lemmas.

\begin{lemma}\label{lem-re}
Let $m\geq 4$ be even and  $q=2^m$. Define
\begin{align*}
   &  R_{(3,i)}=\# \left \{x \in \gf(q)^*: \tr_{q/2}(x)=i~and~  \mbox{x is a cubic residue} \right \}          \nonumber  \\
   & \bar{R}_{(3,i)}=\# \left \{x \in \gf(q)^*: \tr_{q/2}(x)=i~ and~ \mbox{x is not a cubic residue} \right \},
\end{align*}
where $i=0$ or $i=1$. Then
\begin{align*}
\left \{
\begin{array}{l}
 R_{(3,0)}=\frac{1}{6}\left (2^m-2+(-2)^{m/2 +1} \right ), \\
R_{(3,1)}=\frac{2^m-1}{3}-\frac{1}{6} \left (2^m-2+(-2)^{m/2 +1} \right ), \\
\bar{R}_{(3,0)}= (2^{m-1}-1)-\frac{1}{6} \left (2^m-2+(-2)^{m/2 +1} \right ), \\
 \bar{R}_{(3,1)}= \frac{2^m}{3}-\frac{(-2)^{m/2}}{3}.
\end{array}
\right .
\end{align*}
\end{lemma}

\begin{proof}
By definition, we get
\begin{align*}
R_{(3,0)}&=\frac{1}{3} \cdot \# \{x \in \gf(q)^*: \tr_{q/2}(x^3)=0\}   \\
& =\frac{1}{6}\sum_{z\in \gf(2)}\sum_{x\in \gf(q)^*}(-1)^{z\tr_{q/2}(x^3)} \\
& =\frac{1}{6} \left (q-2+\sum_{x\in \gf(q)}(-1)^{\tr_{q/2}(x^3)} \right).
\end{align*}
Then the value of $R_{(3,0)}$ follows from Lemma \ref{expsum}. Note that there are $\frac{q-1}{3}$ cubic residues in $\gf(q)^*$. This gives
$$
\# \{x \in \gf(q)^*: \tr_{q/2}(x)=0 \}=\frac{q}{2}-1.
$$
Then the desired conclusions follow from
\begin{align*}
\left\{  \begin{array}{l}
   R_{(3,0)}+\bar{R}_{(3,0)}=\frac{q}{2}-1,  \\
   R_{(3,0)}+R_{(3,1)}=\frac{q-1}{3},  \\
   \bar{R}_{(3,0)}+\bar{R}_{(3,0)}=\frac{2(q-1)}{3}.
   \end{array}
\right.
\end{align*}
This completes the proof.
\end{proof}

\begin{lemma}\label{lem-even41}
Let $m\geq 4$ be even, $e$ be a positive integer with $\gcd(m,e)=1$, and $q=2^m$. Let $N_{(0,1)}$ be  the number of solutions $(x,y,z,u)\in \gf(q)^4$ of  the system of equations
\begin{align}\label{eqx-even1}
\left\{
  \begin{array}{lllll}
    x&+y&+z&+u&=0,  \\ [2mm]
    x^{2^e+1}&+y^{2^e+1}&+z^{2^e+1}&+u ^{2^e+1}&=1, \\[2mm]
    \multicolumn{4}{c}{\# \left \{x,y,z,u \right \}}&=4.
  \end{array}
\right.
\end{align}
 Then
$N_{(0,1)} =q \left (q-2-(-1)^{m/2}2^{m/2+1} \right ).$
\end{lemma}

\begin{proof}
Let us first observe that (\ref{eqx-even1}) is equivalent to the following system of equations
\begin{align}\label{eqx-even1-1}
\left\{
  \begin{array}{lllll}
    x&+y&+z&+u&=0,  \\ [2mm]
    x^{2^e+1}&+y^{2^e+1}&+z^{2^e+1}&+u ^{2^e+1}&=1.
  \end{array}
\right.
\end{align}
Set $S_j=x^j+y^j+z^j+u^j$, where $x,y,z,u\in \gf(q)$ and $j$ is a positive integer.
An easy computation shows that
\begin{align*}
N_{(0,1)}
&=\frac{1}{q^2} \sum_{a,b, x,y,z,u \in \gf(q)} \chi_1(bS_1)\chi_1 \left (a(S_{2^e+1}-1) \right )   \\
&=\frac{1}{q^2} \sum_{a,b\in \gf(q)}\chi_1(-a) \left(\sum_{x\in \gf(q)}\chi_1\left (ax^{2^e+1}+bx \right ) \right)^4   \\
&=\frac{1}{q^2}  \left (\sum_{a\in \gf(q)}\chi_1(-a) \left  (\sum_{x\in \gf(q)}\chi_1 \left (ax^{2^e+1} \right ) \right )^4 +    \right .     \\
&~~~\left . \sum_{a\in \gf(q)^*}\sum_{b\in \gf(q)^*}\chi_1(-a) \left ( \sum_{x\in \gf(q)}\chi_1\left (ax^{2^e+1}+bx \right ) \right )^4  \right )          \\
&=\frac{1}{q^2}\left  (q^4+ \left (R_{(3,0)}-R_{(3,1)} \right ) 2^{m+4}+\left (\bar{R}_{(3,0)}-\bar{R}_{(3,1)} \right ) 2^{m}   +  \right .     \\
&~~\left . \left (R_{(3,0)}-R_{(3,1)}\right )\left (2^{3m+2}-2^{2m+4} \right )+ \left (\bar{R}_{(3,0)}-\bar{R}_{(3,1)} \right ) (2^m-1)2^{2m} \right ),
\end{align*}
where $R_{(3,0)},R_{(3,1)},\bar{R}_{(3,0)}$ and $\bar{R}_{(3,1)}$ were defined in Lemma \ref{lem-re}  and the last equality holds due to Lemmas \ref{expsum} and \ref{expsum1}. Then  the desired conclusions follow from Lemma \ref{lem-re}.
\end{proof}

\begin{theorem}\label{main-even41}
Let $m\geq 4$ be even and  $e$ be a positive integer with $\gcd(m,e)=1$. Let $q=2^m$,  $f(x)=x^{2^e+1}$ and $\C$ be defined in (\ref{eq:cf}). Let
 $T=\gf(4)=\left \{ 0, 1, w,w^2 \right \} \subseteq \gf(q)$, where $w$ is a generator of $\gf(4)^*$.
 Then the shortened code $\C_{T}$ is a $[2^{m}-4, 2m-3, 2^{m-1}-2^{m/2}]$ binary linear code with the weight distribution in Table \ref{tab-even41}.
\end{theorem}

\begin{table}[ht]
\begin{center}
\caption{The weight distribution of $\C_T$ in Theorem \ref{main-even41} }\label{tab-even41}
\begin{tabular}{cc} \hline
Weight  &  Multiplicity   \\ \hline
$0$          &  $1$ \\ [2mm]
$2^{m-1}-2^{m/2}$    & $ 1/3 \cdot 2^{ m/2-6} \left  (-16 + 2^{3m/2} - 2^{m+1}(-4+(-1)^{m/2}) - 2^{4 + m/2}(-1+(-1)^{m/2}) \right )$ \\ [2mm]
$2^{m-1}-2^{(m-2)/2}$    & $ 1/24 \cdot\left  (2^{m/2+2}+2^m \right ) \left (2^m+(-1)^{m/2}2^{m/2}-2 \right )$ \\ [2mm]
$2^{m-1}$    & $ -1 + 2^{2m-5}-(-1)^{m/2} 2^{ 3m/2 -4}$ \\            [2mm]
$2^{m-1}+2^{(m-2)/2}$    & $ 1/24 \cdot\left  (-2^{m/2+2}+2^m \right  ) \left (2^m+(-1)^{m/2}2^{m/2}-2 \right ) $ \\ [2mm]
$2^{m-1}+2^{m/2}$    & $ 1/3 \cdot 2^{ m/2-6}\left  (16 + 2^{3m/2} - 2^{m+1}(4+(-1)^{m/2}) + 2^{4 + m/2}(1+(-1)^{m/2}) \right )$ \\ [2mm]
\hline
\end{tabular}
\end{center}
\end{table}

\begin{proof}
By Lemma \ref{lem-cf}, $\C$ is a $[2^{m}, 2m+1, 2^{m-1}-2^{m/2}]$ binary code with the weight distribution in Table \ref{tab-cf1} and the minimum distance of  $\C^\bot$ is 6.
Note that  $\lambda _{T,6}(\C^\bot)=0$ by Lemma \ref{main-32} and the definition of $T$. Thus, the minimum distance of  $\left (\mathcal C^{\perp} \right )^{T}$ is at least 3. This means that
\begin{align}\label{eq-4A12}
A_1\left ( \left (\mathcal C^{\perp} \right )^{T}  \right )=A_2\left ( \left (\mathcal C^{\perp} \right )^{T}  \right )=0.~
\end{align}
Further, from Theorem \ref{main-even33} and the definition of $T$, we have
\begin{align}\label{eq:A3}
\lambda _{\hat{T},6}(\C^\bot)=\frac{1}{6}\left (q-2-(-1)^{m/2} 2^{m/2+1}  \right )-1
\end{align}
for any $\hat{T}=\{\hat{x}_1, \hat{x}_2, \hat{x}_3\} \subseteq T$. Therefore,
\begin{align}\label{eq-4A3}
A_3\left ( \left (\mathcal C^{\perp} \right )^{T}  \right )= \binom{4 }{ 3 }\cdot  \left(\frac{1}{6}\left (q-2-(-1)^{\frac{m}{2}} 2^{\frac{m}{2}+1} \right )-1 \right).~
\end{align}
Note that the solutions of the system (\ref{eqx-even1}) have symmetrical property and $(x,y,z,u)=(0,1,w,w^2)$ is a solution of the system (\ref{eqx-even1}). From Lemma \ref{lem-even41} we get
\begin{align}\label{eq-84}
\lambda _{T,8}(\C^\bot)=\frac{N_{(0,1)}}{4!}-1-\binom{4 }{ 3 }\cdot \lambda _{\hat{T},6}(\C^\bot),
\end{align}
where $N_{(0,1)}$ was defined in Lemma \ref{lem-even41} and $\lambda _{\hat{T},6}(\C^\bot)$ was given in (\ref{eq:A3}). By the proof of Lemma \ref{le-even31}, we have
\begin{align}\label{eq-62}
A_4\left ( \left (\mathcal C^{\perp} \right )^{\{\bar{x}_1,\bar{x}_2\}}  \right ) =\frac{2}{3}\cdot (2^{m-2}-1)^2,
\end{align}
where $\{\bar{x}_1,\bar{x}_2\} \subseteq T $. It is obvious that for any $\{\bar{x}_1,\bar{x}_2\} \subseteq T $ there exist only two $3$-subsets $\bar{T}$ of $T$ such that $\{\bar{x}_1,\bar{x}_2\} \subseteq \bar{T} \subseteq T$. By definition we have
\begin{align}\label{eq-4A4}
A_4\left ( \left (\mathcal C^{\perp} \right )^{T}  \right )= \binom{4 }{ 2 }\cdot  \left( A_4\left ( \left (\mathcal C^{\perp} \right )^{\{0,1\}}  \right )-2   \lambda _{\hat{T},6}(\C^\bot)    \right)+
\lambda _{T,8}(\C^\bot),
\end{align}
where $\hat{T}=\{0,1,w\}$.
Combining Equations (\ref{eq:A3}),  (\ref{eq-84}) and (\ref{eq-62}) with Equation (\ref{eq-4A4}) yields
\begin{align}\label{eq-even4A4}
A_4\left ( \left (\mathcal C^{\perp} \right )^{T}  \right )= 4 \left(2^{m-2}-1 \right)^2-\frac{8}{3} \left(  2^m-2-(-1)^{m/2} 2^{m/2+1} \right)+\frac{N_{(0,1)}}{24}+15.
\end{align}

Note that the shortened code $\C_{T}$ has length $2^m-4$ and dimension $2m-3$ due to $\#T=4$ and Lemma \ref{lem:C-S-P}.
By definition and the weight distribution in Table \ref{tab-cf1}, we have that
$A_i\left (\C_T \right )=0$ for $i\not \in \{0, i_1, i_2, i_3,i_4,i_5\} $, where $i_1=2^{m-1}-2^{m/2}$, $i_2=2^{m-1}-2^{(m-2)/2}$
, $i_3=2^{m-1}$, $i_4=2^{m-1}+2^{(m-2)/2}$  and $i_5=2^{m-1}+2^{m/2}$.
Using Lemma \ref{lem:C-S-P} and Equations (\ref{eq-4A12}), (\ref{eq-4A3}) and (\ref{eq-even4A4}) and applying the first five Pless power moments of (\ref{eq:PPM}) yields the weight distribution in Table \ref{tab-even41}. This completes the proof.
\end{proof}

\begin{example}\label{exa-even41}
Let $m=4$ and $e=1$. Then the shortened code $\C_{T}$  in Theorem \ref{main-even41} is a $[12,5,4]$ binary linear code with the weight enumerator $1+3z^{4}+24z^{6}+3z^{8}+z^{12}$. This code $\C_{T}$ is optimal. Its dual $\C_{T}^\perp$ has parameters $[12,7,4]$ and is optimal according to the tables of best known codes   maintained at http://www.codetables.de.
\end{example}

%

\section{Shortened linear codes from PN functions} \label{sec-PN}

In this section, we study some shortened linear codes  from  certain PN functions and determine their parameters.

Let $p$ be odd prime and $q=p^m$. Let $f(x)=x^2$ and $\C$ be defined by (\ref{eq:cf}). Note that we index the coordinators of
the codewords in $\C$ with the elements in $\gf(q)$.  It is known that the code $\C$ is a $[q,2m+1]$ linear code with the the weight distribution given in \cite{LQ2009}. Note that the code $\C$ is affine invariant, and thus holds $2$-designs. Then the following theorem is easily derived from the parameters of the codes $\C$ in \cite{LQ2009} and Theorem \ref{thm:sct-code}, and we omit its proof.

\begin{theorem}\label{main-pdesign2}
Let $p$ be an odd prime, $m$ and $t$ be positive integers. Let $q=p^m$,  $f(x)=x^2$ and $\C$ be defined in (\ref{eq:cf}). Suppose $T$ is a $t$-subset of $\mathcal P(\C):=\gf(q)$.
We have the following results.
\begin{itemize}
  \item It $t=1$, then the shortened code $\C_{T}$ is a $[p^{m}-1, 2m]$ linear code with the weight distribution in Table \ref{tab-px2odd1} (resp., Table \ref{tab-px2even1}) when $m$ is odd (resp., even).
  \item It $t=2$, then the shortened code $\C_{T}$ is a $[p^{m}-2, 2m-1]$ linear code with the weight distribution in Table \ref{tab-px2odd2} (resp., Table \ref{tab-px2even2}) when $m$ is odd (resp., even).
\end{itemize}
\end{theorem}

\begin{table}[ht]
\begin{center}
\caption{The weight distribution of $\C_T$ for $m$ odd and $t=1$ }\label{tab-px2odd1}
\begin{tabular}{cc} \hline
Weight  &  Multiplicity   \\ \hline
$0$          &  $1$ \\ [2mm]
$p^{m-1}(p-1)$    & $  (p^m-1) (1 + p^{m-1})    $ \\ [2mm]
$p^{m-1}(p-1)-p^{\frac{m-1}{2}}$    & $  1/2 \cdot  (p-1) p^{(m-3)/2} (p^m-1) \left (p + p^{(1 + m)/2} \right )   $ \\ [2mm]
$p^{m-1}(p-1)+p^{\frac{m-1}{2}}$    & $  1/2 \cdot  (p-1) p^{(m-3)/2} (p^m-1)\left  (-p + p^{(1 + m)/2} \right )    $ \\ [2mm]
\hline
\end{tabular}
\end{center}
\end{table}

\begin{table}[ht]
\begin{center}
\caption{The weight distribution of $\C_T$ for $m$ odd and $t=2$ }\label{tab-px2odd2}
\begin{tabular}{cc} \hline
Weight  &  Multiplicity   \\ \hline
$0$          &  $1$ \\ [2mm]
$p^{m-1}(p-1)$    & $  p^{2 m-2} -1  $ \\ [2mm]
$p^{m-1}(p-1)-p^{\frac{m-1}{2}}$    & $  ( p-1) \left (-p^{(m-1)/2} + p^{2m-2} + 2 p^{(3m-3)/2} \right ) /2  $ \\ [2mm]
$p^{m-1}(p-1)+p^{\frac{m-1}{2}}$    & $   ( p-1) \left (p^{(m-1)/2} + p^{2m-2} - 2 p^{(3m-3)/2} \right )/2 $ \\ [2mm]
\hline
\end{tabular}
\end{center}
\end{table}

\begin{table}[ht]
\begin{center}
\caption{The weight distribution of $\C_T$ for $m$  even and $t=1$  }\label{tab-px2even1}
\begin{tabular}{cc} \hline
Weight  &  Multiplicity   \\ \hline
$0$          &  $1$ \\
$(p-1)\left (p^{m-1}-p^{(m-2)/2} \right ) $    & $ p^{ m/2-1} \left (p + p^{m/2}-1 \right ) (p^m-1)/2 $  \\
$(p-1) p^{m-1}-p^{(m-2)/2}  $     &    $ ( p-1) p^{m/2-1} \left (p^{m/2}-1 \right ) \left (1 + p^{m/2} \right )^2  /2   $              \\
$(p-1) p^{m-1}$                       & $  p^{m}-1$               \\
$(p-1)p^{m-1}+p^{(m-2)/2}$     &   $ ( p-1) p^{m/2-1} \left (p^{m/2}-1 \right )^2 \left (1 + p^{m/2} \right )  /2   $   \\
$(p-1)(p^{m-1}+p^{(m-2)/2})$    & $ p^{ m/2-1} \left (-p + p^{m/2}+1 \right ) (p^m-1)/2 $ \\
\hline
\end{tabular}
\end{center}
\end{table}

\begin{table}[ht]
\begin{center}
\caption{The weight distribution of $\C_T$ for $m$  even and $t=2$  }\label{tab-px2even2}
\begin{tabular}{cc} \hline
Weight  &  Multiplicity   \\ \hline
$0$          &  $1$ \\
$(p-1) \left (p^{m-1}-p^{(m-2/)2} \right ) $    & $p^{m/2-2} \left (p^{m/2}-1 \right ) \left (-1 + p + p^{m/2} \right ) \left (p + p^{m/2} \right )/2  $  \\
$(p-1) p^{m-1}-p^{(m-2)/2}  $     &    $ (p-1) p^{m/2-2} \left (1 + p^{m/2} \right ) \left (-p + p^{m/2} + p^m \right )/2   $              \\
$(p-1) p^{m-1}$                       & $  p^{m-1}-1$               \\
$(p-1)p^{m-1}+p^{(m-2)/2}$     &   $ (p-1) p^{m/2-2}\left  (-1 + p^{m/2} \right ) \left (-p - p^{m/2} + p^m \right )/2   $   \\
$(p-1)\left (p^{m-1}+p^{(m-2)/2} \right )$    &  $p^{m/2-2} \left (p^{m/2}+1 \right ) \left (1 - p + p^{m/2} \right ) \left (p^{m/2}-p\right )/2  $  \\
\hline
\end{tabular}
\end{center}
\end{table}

\begin{example}\label{exa-p1}
Let $m=3$, $p=3$ and $T$ be a $1$-subset of $\mathcal P(\C)$. Then the shortened code $\C_{T}$  in Theorem \ref{main-pdesign2} is a $[26, 6, 15]$ linear code with the weight enumerator $1+312 z^{15}+260 z^{18}+ 156 z^{21}$. This code $\C_{T}$ is optimal. Its dual $\C_{T}^\perp$ has parameters $[26, 20, 4]$ and is optimal according to the tables of best known codes   maintained at http://www.codetables.de.
\end{example}

\begin{example}\label{exa-p5}
Let $m=4$, $p=3$ and $T$ be a $1$-subset of $\mathcal P(\C)$. Then the shortened code $\C_{T}$  in Theorem \ref{main-pdesign2} is a $[80, 8, 48]$ linear code with the weight enumerator $1+1320 z^{48}+2400 z^{51}+ 80 z^{54}+1920 z^{57}+840 z^{60}$. This code $\C_{T}$ is optimal. Its dual $\C_{T}^\perp$ has parameters $[80, 72, 4]$ and is optimal according to the tables of best known codes   maintained at http://www.codetables.de.
\end{example}


\begin{example}\label{exa-p3}
Let $m=5$, $p=3$ and $T$ be a $2$-subset of $\mathcal P(\C)$. Then the shortened code $\C_{T}$  in Theorem \ref{main-pdesign2} is a $[241, 9, 153]$ linear code with the weight enumerator $1+8010 z^{153}+6560 z^{162}+ 5112 z^{171}$. This code $\C_{T}$ is optimal. Its dual $\C_{T}^\perp$ has parameters $[241, 232, 3]$ and is optimal according to the tables of best known codes   maintained at http://www.codetables.de.
\end{example}

\begin{example}\label{exa-p4}
Let $m=4$, $p=3$ and $T$ be a $2$-subset of $\mathcal P(\C)$. Then the shortened code $\C_{T}$  in Theorem \ref{main-pdesign2} is a $[79, 7, 48]$ linear code with the weight enumerator $1+528 z^{48}+870 z^{51}+ 26 z^{54}+552 z^{57}+210 z^{60}$. This code $\C_{T}$ is almost optimal. Its dual $\C_{T}^\perp$ has parameters $[79,72,3]$ and is almost optimal according to the tables of best known codes   maintained at http://www.codetables.de.
\end{example}

In the following, we will consider the shortened code $\C_{T}$ of $\C$ for the case $T=\gf(p)$. In order to determine the parameters of $\C_{T}$, we need the next lemmas.

\begin{lemma}\cite{HD2019}\label{lem-equalities}
Let $q=p^m$ with $p$ an odd prime. Then
 \begin{eqnarray*}
\lefteqn{\sharp \{a\in \gf(q)^*:\eta(a)=1\mbox{ and }\tr_{q/p}(a)=0\} } \\
&=&\left\{
\begin{array}{cl}
\frac{p^{m-1}-1-(p-1)p^{\frac{m-2}{2}}(\sqrt{-1})^{\frac{(p-1)m}{2}}}{2}, & \mbox{ if $m$ is even,} \\
\frac{p^{m-1}-1}{2}, & \mbox{ if $m$ is odd.}
\end{array}\right.
\end{eqnarray*}
and
\begin{eqnarray*}
\lefteqn{ \sharp \{a\in \gf(q)^*:\eta(a)=-1\mbox{ and }\tr_{q/p}(a)=0\} } \\
&=&\left\{
\begin{array}{cl}
\frac{p^{m-1}-1+(p-1)p^{\frac{m-2}{2}}(\sqrt{-1})^{\frac{(p-1)m}{2}}}{2}, & \mbox{ if $m$ is even,}\\
\frac{p^{m-1}-1}{2}, & \mbox{ if $m$ is odd.}
\end{array}\right.
\end{eqnarray*}
\end{lemma}

\begin{lemma}\cite{HD2019} \label{lem-equalities1}
Let $q=p^m$ with $p$ an odd prime and $(a,b)\in \gf(q)^2$. Denote

\begin{eqnarray}\label{eqn-50}
N_0(a,b)=\sharp\{x\in \gf(q):\tr_{q/p}(ax^2+bx)=0\}.
\end{eqnarray}
Then the following results follow.
\begin{itemize}
  \item If $m$ is odd, then
  \begin{eqnarray}\label{eqn-5}
\lefteqn{ N_0(a,b)=} \nonumber \\
&
\left\{\begin{array}{cl}
p^m, & \mbox{ if }(a,b)=(0,0),\\
p^{m-1}, & \mbox{ if }a=0,\ b\neq 0,\mbox{ or }a\neq 0,\ \tr_{q/p}(\frac{b^2}{4a})=0,\\
p^{m-1}+p^{\frac{m-1}{2}}(-1)^{\frac{(p-1)(m+1)}{4}}, &
\mbox{ if }a\neq 0,\ \tr_{q/p}(\frac{b^2}{4a})\neq 0,\ \eta(a)\bar{\eta}\left(-\tr_{q/p}(\frac{b^2}{4a})\right)=1,\\
p^{m-1}+p^{\frac{m-1}{2}}(-1)^{\frac{(p-1)(m+1)+4}{4}}, &
\mbox{ if }a\neq 0,\ \tr_{q/p}(\frac{b^2}{4a})\neq 0,\ \eta(a)\bar{\eta}\left(-\tr_{q/p}(\frac{b^2}{4a})\right)=-1.
\end{array} \right.
\end{eqnarray}
  \item If $m$ is even, then
  \begin{eqnarray}\label{eqn-6}
N_0(a,b)=
\left\{\begin{array}{cl}
p^m & \mbox{ if }(a,b)=(0,0),\\
p^{m-1}, & \mbox{ if }a=0,\ b\neq 0,\\
p^{m-1}+(p-1)p^{\frac{m-2}{2}}(-1)^{\frac{m(p-1)+4}{4}}, & \mbox{ if }a\neq 0,\ \tr_{q/p}(\frac{b^2}{4a})=0,\ \eta(a)=1,\\
p^{m-1}+(p-1)p^{\frac{m-2}{2}}(-1)^{\frac{m(p-1)}{4}}, & \mbox{ if }a\neq 0,\ \tr_{q/p}(\frac{b^2}{4a})=0,\ \eta(a)=-1,\\
p^{m-1}+p^{\frac{m-2}{2}}(-1)^{\frac{m(p-1)}{4}}, &
\mbox{ if }a\neq 0,\ \tr_{q/p}(\frac{b^2}{4a})\neq 0,\ \eta(a)=1,\\
p^{m-1}+p^{\frac{m-2}{2}}(-1)^{\frac{m(p-1)+4}{4}}, &
\mbox{ if }a\neq 0,\ \tr_{q/p}(\frac{b^2}{4a})\neq 0,\ \eta(a)=-1.
\end{array} \right.
\end{eqnarray}
\end{itemize}
\end{lemma}

\begin{lemma}\label{lem-sign}
Let $q=p^m$ with $p$ an odd prime and $a\in \gf(q)^*$. Define $f(x)=\tr_{q/p}(-\frac{1}{4a}x^2)$. Then the dual of $f(x)$ is $f^*(x)=\tr_{q/p}(ax^2)$ and the sign of the Walsh transform of $f(x)$ is
$$ \varepsilon=\eta(a)(-1)^{m-1+\frac{q-1}{2}}.$$
\end{lemma}

\begin{proof}
Note that $f(x)$ is a weakly regular bent function. By definition and Lemma \ref{lem-32A2},
we have
$$
\mathcal{W}_{f}(\beta)=
\sum_{x\in\gf(q)}\zeta_p^{\tr_{q/p}(-\frac{1}{4a}x^2)
-\tr_{q/p}(\beta x)}
=\zeta_p^{\tr_{q/p}(a \beta^2)} \cdot \eta\left (-\frac{1}{4a} \right ) G(\eta, \chi) =\zeta_p^{f^*(\beta)} \cdot  \varepsilon \cdot  \sqrt{p^*}^m.
$$
Then the desired conclusions follow from Lemma \ref{lem-32A1}.
\end{proof}

\begin{lemma}\label{lem-psign}
Let $q=p^m$ with $p$ an odd prime, $a\in \gf(q)^*$ and $e$ be any positive integer such that $m/\gcd(m,e)$ is odd. Define $f(x)=\tr_{q/p}(ax^{p^e+1})$. Then the dual of $f(x)$ is $f^*(\beta )=\tr_{q/p}\left (-a \left (x_{a,-\beta} \right ) ^{p^e+1} \right )$  and the sign of the Walsh transform of $f(x)$ is
\begin{align*}
\varepsilon =
\left\{
  \begin{array}{ll}
    (-1)^{m-1+\frac{q-1}{2}} \eta(a) ,        & \mbox{if}~ p\equiv 1 ~\bmod~4, \\
    (-1)^{\frac{q-1}{2}+1} \eta(a) ,  & \mbox{if}~ p\equiv 3 ~\bmod~4, \\
  \end{array}
\right.
\end{align*}
where $\beta \in \gf(q)$ and $x_{a,-\beta}$ is the unique solution of the equation
$$
a^{p^e}x^{p^{2e}}+a x+ (-\beta)^{p^e}=0.
$$
\end{lemma}

\begin{proof}
The proof is similar to that of Lemma \ref{lem-sign}, and we omit it here. Note that  $f(x)$ is a weakly regular bent function. Then the desired conclusions follow from the definitions and Lemma \ref{lem-psum}.
\end{proof}

\begin{lemma}\label{lem-psign1}
Let $q=p^m$ with $p$ an odd prime, $a\in \gf(q)^*$, $\gamma \in \gf(p)$ and $e$ be any positive integer such that $m/\gcd(m,e)$ is odd. Define
$$
\tilde{N}_{\gamma}=\# \left \{b\in \gf(q): \tr_{q/p}(b)=0  ~and~ \tr_{q/p}\left (ax_{a,b} ^{p^e+1} \right )=\gamma \right \}
$$
where $x_{a,b}$ is the unique solution of the equation
$
a^{p^e}x^{p^{2e}}+a x+ b^{p^e}=0.
$
If $\tr_{q/p}(a)=0$, then
$$
\tilde{N}_{\gamma} =
\left\{
  \begin{array}{cl}
    p^{m-2}+\varepsilon \bar{\eta}^{m/2}(-1)(p-1)p^{(m-2)/2},                        & \hbox{if $\gamma=0$ and $m$ is even,} \\
    p^{m-2},                                                                         & \hbox{if $\gamma=0$ and $m$ is odd,} \\
    \frac{p-1}{2}\left(p^{m-2}-\varepsilon \bar{\eta}^{m/2}(-1)p^{(m-2)/2}\right) ,   & \hbox{if $\gamma \neq 0$ and $m$ is even,} \\
    \frac{p-1}{2}\left( p^{m-2}+\varepsilon \sqrt{p*}^{m-1} \right),                & \hbox{if $\gamma \in \textup{SQ}$ and $m$ is odd,} \\
    \frac{p-1}{2}\left( p^{m-2}-\varepsilon \sqrt{p*}^{m-1} \right).                 & \hbox{if $\gamma \in \textup{NSQ}$ and $m$ is odd,} \\
  \end{array}
\right.
$$
where $\varepsilon$ was given in Lemma \ref{lem-psign}.
\end{lemma}

\begin{proof}
By the definition of $x_{a,b}$, we have
$$a^{p^e} x_{a,b}^{p^{2e}}+a x_{a,b}+ b^{p^e}=\left (a^{p^{-e}} x_{a,b} \right )^{p^{2e}}+a x_{a,b}+ b^{p^e}=0.$$
This gives
$$\tr_{q/p}\left((a^{p^{-e}}+a) x_{a,b} +b\right)=0.$$
Since $a^{p^e} x^{p^{2e}}+a x$ is a linear permutation polynimial over $\gf(q)$, we have
\begin{eqnarray}\label{eqn-61}
\tilde{N}_{\gamma}=\# \left  \{x\in \gf(q): \tr_{q/p}(a x ^{p^e+1})=\gamma ~and~\tr_{q/p}\left((a^{p^{-e}}+a) x  \right)=0 \right \}.
\end{eqnarray}
Note that $x=1$ is the unique solution of the equation $$a^{p^e}(x)^{p^{2e}}+a x+ (-a^{p^{-e}}-a)^{p^e}=0 .$$
By Lemma \ref{lem-psign}, we get
\begin{eqnarray}\label{eqn-62}
f^*(a^{p^{-e}}+a)=\tr_{q/p}(-a)=0
\end{eqnarray}
 where $f^*$ is the dual of $\tr_{q/p}(ax^{p^e+1})$. By Equations (\ref{eqn-61}) and (\ref{eqn-62}), the desired conclusions follow from Lemmas \ref{2lem11}, \ref{2lem14} and \ref{lem-psign}.
\end{proof}

\begin{theorem}\label{main-51}
Let $p$ be an odd prime and $m$ be a positive integer. Let $q=p^m$,  $f(x)=x^2$ and $\C$ be defined in (\ref{eq:cf}).
Let $T=\gf(p)$. Then the shortened code $\C_{T}$ is a $[p^{m}-p, 2m-2]$ linear code. If $m$ is odd, the weight distribution of $\C_{T}$ is given in Table \ref{tab-51}, where $B=(-1)^{\frac{q-1}{2}+\frac{(p-1)(m-1)}{4}} p^{(m-1)/2}$; if $m\geq 2$ is even, the weight distribution of $\C_{T}$ is given in Table \ref{tab-52}, where
$B_1=\frac{p^{m-1}-1-(p-1)p^{\frac{m-2}{2}}(\sqrt{-1})^{\frac{(p-1)m}{2}}}{2}$
and
$B_2= (-1)^{\frac{q+1}{2} + \frac{m(p-1)}{4}}  (p-1)p^{(m-2)/2}$.
\end{theorem}

\begin{table}[ht]
\begin{center}
\caption{The weight distribution of $\C_T$ for $m$  odd }\label{tab-51}
\begin{tabular}{cc} \hline
Weight  &  Multiplicity   \\ \hline
$0$          &  $1$ \\ [2mm]
$p^{m-1}(p-1)$    & $ (p^{m-1}-1)(p^{m-2}+1)$ \\ [2mm]
$p^{m-1}(p-1)-p^{\frac{m-1}{2}}(-1)^{\frac{(p-1)(m+1)}{4}}$    & $ \frac{p^{m-1}-1}{2}\cdot (p-1)(p^{m-2}+B)$ \\ [2mm]
$p^{m-1}(p-1)+p^{\frac{m-1}{2}}(-1)^{\frac{(p-1)(m+1)}{4}}$    & $ \frac{p^{m-1}-1}{2}\cdot (p-1)(p^{m-2}-B)$ \\ [2mm]
\hline
\end{tabular}
\end{center}
\end{table}

\begin{table}[ht]
\begin{center}
\caption{The weight distribution of $\C_T$ for $m$  even  }\label{tab-52}
\begin{tabular}{cc} \hline
Weight  &  Multiplicity   \\ \hline
$0$          &  $1$ \\
$p^{m-1}(p-1)$    & $  p^{m-1}-1$ \\
$p^{m-1}(p-1)-(p-1)p^{\frac{m-2}{2}}(-1)^{\frac{m(p-1)+4}{4}}$    & $ B_1\cdot (p^{m-2}+B_2) $ \\
$p^{m-1}(p-1)-(p-1)p^{\frac{m-2}{2}}(-1)^{\frac{m(p-1)}{4}}$    & $ (p^{m-1}-1-B_1)\cdot (p^{m-2}-B_2)$ \\
$p^{m-1}(p-1)-p^{\frac{m-2}{2}}(-1)^{\frac{m(p-1)}{4}}   $     &    $  B_1\cdot (p^{m-1}-p^{m-2}-B_2) $ \\
$p^{m-1}(p-1)-p^{\frac{m-2}{2}}(-1)^{\frac{m(p-1)+4}{4}}$     &   $ (p^{m-1}-1-B_1)\cdot (p^{m-1}-p^{m-2}+B_2)  $   \\
\hline
\end{tabular}
\end{center}
\end{table}

\begin{proof}
Denote
\begin{align}\label{eq-H}
H=\left\{(a,b) \in \gf(q)^2: \tr_{q/p}(a)=\tr_{q/p}(b)=0 \right\}.
\end{align}
By definition, the shortened code $\C_{T}$ has length $p^m-p$. The desired conclusion on the dimension of $\C_{T}$
will be clear after the weight distribution of $\C_{T}$ is settled below.

Since $T=\gf(p)$, the weight distribution of $\C_{T}$ is the same as the subcode
\begin{align}\label{eq:cf1}
{\C}(T)=\left\{\left(\tr_{q/p}(ax^2+bx)\right)_{x \in \gf(q)}:(a,b)\in H \right\}.
\end{align}
For each $(a,b) \in H $, define the corresponding codeword
\begin{eqnarray*}\label{eqn-mcodeword}
\bc (a,b)=\left(\tr_{q/p}(ax^2+bx)\right)_{x\in \gf(q)} \in {\C}(T).
\end{eqnarray*}
Then the Hamming weight of $\bc (a,b)$ is
\begin{eqnarray*}
\wt(\bc(a,b))&=& q-N_0(a,b)
\end{eqnarray*}
 where $N_0(a,b)$ was defined in Equation (\ref{eqn-50}). We discuss the value of $\wt(\bc(a,b))$ in the following two cases.

(\uppercase\expandafter{\romannumeral1}) The case that $m$ is odd.

From Equation (\ref{eqn-5}), we get
\begin{eqnarray*}
& & \wt(\bc(a,b))=q-N_0(a,b)  \nonumber \\
& & =
\left\{\begin{array}{ll}
0, & \mbox{ if }(a,b)=(0,0) \\
p^{m-1}(p-1), & \mbox{ if }a=0~and~\ b\neq 0, \mbox{ or }a\neq 0 ~ and ~ \tr_{q/p}(\frac{b^2}{4a})=0\\
p^{m-1}(p-1)-p^{\frac{m-1}{2}}(-1)^{\frac{(p-1)(m+1)}{4}}, &  \mbox{ if }a\neq 0,\tr_{q/p}(\frac{b^2}{4a})\neq 0,\ \eta(a)\bar{\eta}\left(-\tr_{q/p}(\frac{b^2}{4a})\right)=1\\
p^{m-1}(p-1)+p^{\frac{m-1}{2}}(-1)^{\frac{(p-1)(m+1)}{4}},  &    \mbox{ if }a\neq 0, \tr_{q/p}(\frac{b^2}{4a})\neq 0,\ \eta(a)\bar{\eta} \left(-\tr_{q/p}(\frac{b^2}{4a})\right)=-1 \\
~ &
\end{array} \right. \\ [2mm]
& & = \left\{\begin{array}{ll}
0,    &   \mbox{ with $1$ time},\\
p^{m-1}(p-1),    &   \mbox{ with $(p^{m-1}-1)(p^{m-2}+1)$  times},\\
p^{m-1}(p-1)-p^{\frac{m-1}{2}}(-1)^{\frac{(p-1)(m+1)}{4}},    &   \mbox{ with $ \frac{p^{m-1}-1}{2}\cdot (p-1)(p^{m-2}+B)$  times},\\
p^{m-1}(p-1)+p^{\frac{m-1}{2}}(-1)^{\frac{(p-1)(m+1)}{4}},   &   \mbox{ with $\frac{p^{m-1}-1}{2}\cdot (p-1)(p^{m-2}-B)$   times}, \\
\end{array} \right.
\end{eqnarray*}
 when $(a,b)$ runs through $H$, where
$$B=(-1)^{m-1+\frac{q-1}{2}} \cdot \sqrt{p^*}^{m-1}=(-1)^{\frac{q-1}{2}+\frac{(p-1)(m-1)}{4}} p^{(m-1)/2}$$
and the frequency is obtained from Lemmas
 \ref{2lem11}, \ref{2lem14}, \ref{lem-equalities} and \ref{lem-sign}.
We first compute the frequency $A_w$ of the nonzero weight $w$, where
$$w=p^{m-1}(p-1)-p^{\frac{m-1}{2}}(-1)^{\frac{(p-1)(m+1)}{4}}.$$
Then
$$A_w=\sharp \{(a,b)\in H:a\neq 0,\ \tr_{q/p}\left (\frac{b^2}{4a} \right ) \neq 0 ,\ \eta(a)\bar{\eta} \left(-\tr_{q/p}\left (\frac{b^2}{4a} \right )\right)=1\}.$$
Clearly, the number of $a \in \gf(q)^*$ such that $\tr_{q/p}(a)=0$ and $\eta(a)=1$, is $\bar{n}_a=\frac{p^{m-1}-1}{2}$ by Lemma \ref{lem-equalities}; if we fix $a$ with $\tr_{q/p}(a)=0$ and $\eta(a)=1$, the number of $b$ such that $\tr_{q/p}(b)=0$  and $ \bar{\eta} (-\tr_{q/p}(\frac{b^2}{4a}))=1$, is $\bar{n}_b=\frac{p-1}{2}(p^{m-2}+  (-1)^{m-1+\frac{q-1}{2}} \cdot \sqrt{p^*}^{m-1})$ by Lemmas \ref{2lem14} and \ref{lem-sign}. Meanwhile, the number of $a \in \gf(q)^*$ such that $\tr_{q/p}(a)=0$ and $\eta(a)=-1$, is $\hat{n}_a=\frac{p^{m-1}-1}{2}$ by Lemma \ref{lem-equalities}; if we fix $a$ with $\tr_{q/p}(a)=0$ and $\eta(a)=-1$, the number of $b$ such that $\tr_{q/p}(b)=0$  and $ \bar{\eta} (-\tr_{q/p}(\frac{b^2}{4a}))=-1$, is $\hat{n}_b=\frac{p-1}{2}(p^{m-2}+ (-1)^{m-1+\frac{q-1}{2}} \cdot \sqrt{p^*}^{m-1})$ by Lemmas \ref{2lem14} and \ref{lem-sign}. Hence,
$$A_w=\bar{n}_a \bar{n}_b+\hat{n}_a \hat{n}_b= \frac{p^{m-1}-1}{2} (p-1) \left (p^{m-2}+ (-1)^{m-1+\frac{q-1}{2}} \cdot \sqrt{p^*}^{m-1} \right ).$$
The frequencies of other nonzero weights can be similarly derived.

(\uppercase\expandafter{\romannumeral2}) The case that $m$ is even.

From Equation (\ref{eqn-6}), we get
\begin{eqnarray*}
& & \wt(\bc(a,b))=q-N_0(a,b)  \nonumber \\
& & =
\left\{\begin{array}{ll}
0, & \mbox{ if }(a,b)=(0,0) \\ [2mm]
p^{m-1}(p-1), & \mbox{ if }a=0,\ b\neq 0 \\ [2mm]
p^{m-1}(p-1)-(p-1)p^{\frac{m-2}{2}}(-1)^{\frac{m(p-1)+4}{4}}, & \mbox{ if }a\neq 0,\ \tr_{q/p}(\frac{b^2}{4a})=0,\ \eta(a)=1 \\

p^{m-1}(p-1)-(p-1)p^{\frac{m-2}{2}}(-1)^{\frac{m(p-1)}{4}}, & \mbox{ if }a\neq 0,\ \tr_{q/p}(\frac{b^2}{4a})=0,\ \eta(a)=-1 \\

p^{m-1}(p-1)-p^{\frac{m-2}{2}}(-1)^{\frac{m(p-1)}{4}},  &      \mbox{ if }a\neq 0,\ \tr_{q/p}(\frac{b^2}{4a})\neq 0,\ \eta(a)=1 \\

p^{m-1}(p-1)-p^{\frac{m-2}{2}}(-1)^{\frac{m(p-1)+4}{4}},    &   \mbox{ if }a\neq 0,\ \tr_{q/p}(\frac{b^2}{4a})\neq 0,\ \eta(a)=-1 \\
\end{array} \right. \\ [2mm]
& & =
\left\{\begin{array}{ll}
0, &                                              \mbox{ with $  1 $ time},\\
p^{m-1}(p-1),                                    & \mbox{ with $  p^{m-1}-1 $ times},\\
p^{m-1}(p-1)+(p-1)p^{\frac{m-2}{2}}(-1)^{\frac{m(p-1)}{4}}, & \mbox{ with $  B_1\cdot (p^{m-2}+B_2) $ times}, \\
p^{m-1}(p-1)-(p-1)p^{\frac{m-2}{2}}(-1)^{\frac{m(p-1)}{4}},   & \mbox{ with $  (p^{m-1}-1-B_1)\cdot (p^{m-2}-B_2)  $ times}, \\
p^{m-1}(p-1)-p^{\frac{m-2}{2}}(-1)^{\frac{m(p-1)}{4}},        &       \mbox{ with $  B_1\cdot (p^{m-1}-p^{m-2}-B_2) $ times}, \\
p^{m-1}(p-1)+p^{\frac{m-2}{2}}(-1)^{\frac{m(p-1)}{4}},     & \mbox{ with $ (p^{m-1}-1-B_1)\cdot (p^{m-1}-p^{m-2}+B_2)  $ times}, \\
\end{array} \right.
\end{eqnarray*}
where $(a,b)$ runs through $H$, $B_1=\frac{p^{m-1}-1-(p-1)p^{\frac{m-2}{2}}(\sqrt{-1})^{\frac{(p-1)m}{2}}}{2}$,
$$B_2= (-1)^{m-1+\frac{q-1}{2}} \bar{\eta}^{m/2}(-1) \cdot (p-1)p^{(m-2)/2}=(-1)^{\frac{q+1}{2} + \frac{m(p-1)}{4}}  (p-1)p^{(m-2)/2},$$
and
the frequency is easy to obtain from Lemmas  \ref{2lem11}, \ref{2lem14}, \ref{lem-equalities} and \ref{lem-sign}.
The weight distribution of $\C_{T}$ can be handled in much the same way as the case that $m$ is odd. Details are omitted here.

By the above two cases, the weight distributions in Tables \ref{tab-51} and \ref{tab-52} follow. This completes the proof.
\end{proof}

\begin{example}\label{exa-51}
Let $m=3$ and $p=3$. Then the shortened code $\C_{T}$  in Theorem \ref{main-51} is a $[24,4,15]$ linear code with the weight enumerator $1+48z^{15}+32z^{18}$. This code $\C_{T}$ is optimal. Its dual $\C_{T}^\perp$ has parameters $[24,20,3]$ and is optimal according to the tables of best known codes   maintained at http://www.codetables.de.
\end{example}

\begin{example}\label{exa-53}
Let $m=4$ and $p=3$. Then the shortened code $\C_{T}$  in Theorem \ref{main-51} is a $[78,6,48]$ linear code with the weight enumerator $1+240z^{48}+240z^{51}+26z^{54}+192 z^{57}+30 z^{60}$. This code $\C_{T}$ is almost optimal. Its dual $\C_{T}^\perp$ has parameters $[78,72,2]$ and is almost optimal according to the tables of best known codes   maintained at http://www.codetables.de.
\end{example}

%

\begin{theorem}\label{main-52}
Let $p$ be an odd prime, $m$ and $e$ be positive integers such that $m/\gcd(m,e)$ is odd. Let $q=p^m$,  $f(x)=x^{p^e+1}$ and $\C$ be defined in (\ref{eq:cf}).
Let $T=\gf(p)$. Then the parameters of the shortened code $\C_{T}$ are the same as that of $\C_{T}$ in Theorem \ref{main-51}.
\end{theorem}
\begin{proof}
The proof is similar to that of Theorem \ref{main-51}. Recall that the code $\C$ has length $q$ and dimension $2m+1$. By  definition,   the shortened code $\C_{T}$ has length $p^m-p$. The desired conclusion on the dimension of  $\C_{T}$ will be
clear after the weight distribution of  $\C_{T}$ is settled below. Since  $T=\gf(p)$, the  weight distribution of $\C_{T}$ is the same as the code
\begin{align}\label{eq:cf2}
{\C}(T)=\left\{\left(\tr_{q/p}(ax^{p^e+1}+bx)\right)_{x \in \gf(q)}:(a,b)\in H \right\},
\end{align}
where $H$ was defined by (\ref{eq-H}).

For each $(a,b) \in H$, define the corresponding codeword
\begin{eqnarray*}
\bc (a,b)=\left(\tr_{q/p}(ax^{p^e+1}+bx)\right)_{x\in \gf(q)} \in \hat{\C}.
\end{eqnarray*}
Then the Hamming weight of $\bc (a,b)$ is
\begin{eqnarray}\label{eq-word}
\wt(\bc(a,b))=q-\hat{N}_0(a,b),
\end{eqnarray}
where $\hat{N}_0(a,b)$ was defined in Lemma \ref{lem-NN0}.

We determine the value of $\wt(\bc(a,b))$ and its frequencies according to the parity of $m$ and the residue of $p$ modulo $4$ as follows.
\begin{itemize}
  \item (\uppercase\expandafter{\romannumeral1}) $m$ is odd and $p\equiv 1~\bmod~4$.
  \item (\uppercase\expandafter{\romannumeral2}) $m$ is odd and $p\equiv 3~\bmod~4$.
  \item (\uppercase\expandafter{\romannumeral3}) $m$ is even and $p\equiv 1~\bmod~4$.
  \item (\uppercase\expandafter{\romannumeral4}) $m$ is even and $p\equiv 3~\bmod~4$.
\end{itemize}

Next we only give the proof for the case (\uppercase\expandafter{\romannumeral1})  and omit the proofs for the other there cases whose proofs are similar.

Suppose  that $m$ is odd and $p\equiv 1~\bmod~4$. From Equation (\ref{eq-word}) and Lemma \ref{lem-NN0}, we get
\begin{eqnarray*}
& & \wt(\bc(a,b))=q-\hat{N}_0(a,b)  \nonumber \\
& & =
\left\{\begin{array}{ll}
0, & \mbox{ if }(a,b)=(0,0) \\
p^{m-1}(p-1), & \mbox{ if }ab=0, (a,b)\neq (0,0) \\
               ~          & \mbox{ or }ab \neq 0, \tr_{q/p}(a(x_{a,b})^{p^e+1})= 0 \\
p^{m-1}(p-1)-p^{m/2-1} \sqrt{p^*}                                &  \mbox{ if }ab\neq 0, \tr_{q/p}(a(x_{a,b})^{p^e+1})\neq 0,\eta(a)\eta( \tr_{q/p}(a(x_{a,b})^{p^e+1}) )=1 \\

p^{m-1}(p-1)+p^{m/2-1} \sqrt{p^*}                                &  \mbox{ if }ab\neq 0,  \tr_{q/p}(a(x_{a,b})^{p^e+1})\neq 0, \eta(a)\eta( \tr_{q/p}(a(x_{a,b})^{p^e+1}) )=-1         \\
\end{array} \right. \\ [2mm]
& & = \left\{\begin{array}{ll}
0,    &   \mbox{ with $1$ time},\\
p^{m-1}(p-1),    &   \mbox{ with $(p^{m-1}-1)(p^{m-2}+1)$  times},\\
p^{m-1}(p-1)-p^{\frac{m-1}{2}},    &   \mbox{ with $ \frac{p^{m-1}-1}{2}\cdot (p-1)(p^{m-2}+B)$  times},\\
p^{m-1}(p-1)+p^{\frac{m-1}{2}},  &   \mbox{ with $\frac{p^{m-1}-1}{2}\cdot (p-1)(p^{m-2}-B)$   times}, \\
\end{array} \right.
\end{eqnarray*}
when $(a,b)$ runs through $H$, where $B=(-1)^{m-1+\frac{q-1}{2}} \cdot \sqrt{p^*}^{m-1}$ and the frequency is obtained by Lemmas \ref{lem-equalities} and \ref{lem-psign} .
As an example, we just compute the frequency $A_w$ of the nonzero weight $w$, where
$$w=p^{m-1}(p-1)-p^{\frac{m-1}{2}}. $$
Thus
$$A_w=\sharp \{(a,b)\in H:~ab\neq 0, \tr_{q/p}(a(x_{a,b})^{p^e+1})\neq 0,  \eta(a)\eta( \tr_{q/p}(a(x_{a,b})^{p^e+1}) )=1 \}.$$
Clearly, the number of $a \in \gf(q)^*$ such that $\tr_{q/p}(a)=0$ and $\eta(a)=1$, is $\bar{n}_a=\frac{p^{m-1}-1}{2}$ by Lemma \ref{lem-equalities}; if we fix $a$ with $\tr_{q/p}(a)=0$ and $\eta(a)=1$, the number of $b$ such that $\tr_{q/p}(b)=0$   and $\eta( \tr_{q/p}(a(x_{a,b})^{p^e+1}) )=1$, is $\bar{n}_b=\frac{p-1}{2}\left( p^{m-2}+(-1)^{m-1+\frac{q-1}{2}}\cdot \sqrt{p*}^{m-1} \right)$
by Lemma \ref{lem-psign1}. Meanwhile, the number of $a \in \gf(q)^*$ such that $\tr_{q/p}(a)=0$ and $\eta(a)=-1$, is $\hat{n}_a=\frac{p^{m-1}-1}{2}$ by Lemma \ref{lem-equalities}; if we fix $a$ with $\tr_{q/p}(a)=0$ and $\eta(a)=-1$, the number of nonzero $b$ such that $\tr_{q/p}(b)=0$   and $\eta( \tr_{q/p}(a(x_{a,b})^{p^e+1}) )=-1$, is $\hat{n}_b=\frac{p-1}{2}\left( p^{m-2}+(-1)^{m-1+\frac{q-1}{2}} \sqrt{p*}^{m-1} \right)$  by Lemma \ref{lem-psign1}. Hence,
$$A_w=\bar{n}_a \bar{n}_b+\hat{n}_a \hat{n}_b= \frac{p^{m-1}-1}{2}\cdot (p-1)\left (p^{m-2}+ (-1)^{m-1+\frac{q-1}{2}} \cdot \sqrt{p^*}^{m-1} \right ).$$
The frequencies of other nonzero weights can be similarly derived. This completes the proof of the weight distribution of Table in  \ref{tab-51} for the case $m$   odd and $p\equiv 1~\bmod~4$.

The proofs of the other three cases are similar. The desired conclusions follow from Equation (\ref{eq-word}), Lemmas \ref{lem-NN0}, \ref{lem-equalities} and \ref{lem-psign1}. This completes the proof.
\end{proof}

\section{Concluding remarks}\label{sec-summary}

In this paper,  we mainly investigated some shortened codes of linear codes from PN and APN functions and determined their parameters. The obtained codes have a few weights and many of these codes are optimal or almost optimal. Specifically, the main contributions are summarized below.
\begin{itemize}
  \item For any binary linear code $\C$ with length $q=2^m$ and the weight distribution in Table \ref{tab-cf},
  we determined the weight distributions of the shortened codes $\C_T$ for $\#T \in \{1, 2,3\}$ (see Theorem \ref{main-design1}) and gave a general result on the shortened codes $\C_{T}$ with $\#T=4$ in Theorem \ref{main-31}. Meanwhile, when $m$ is odd, the parameters of the shortened codes $\C_{T}$ of a class of binary linear codes from APN functions were determined in Theorem \ref{main-odd4}.

  \item For any binary linear code $\C$ with length $q=2^m$ and the weight distribution in Table \ref{tab-cf1}, we
  settled the weight distributions of the shortened codes $\C_T$ with $\#T \in \{1,2\}$ (see Theorem \ref{main-design2})
  and developed a general result on the shortened codes $\C_{T}$ with $\#T=3$ in Theorem \ref{main-even32}. Further, the parameters of the shortened codes $\C_{T}$  from certain APN functions were determined for $\#T=3$ and $\#T=4$  in Theorems \ref{main-even33} and \ref{main-even41}.
  \item Two classes of $p$-ary shortened codes $\C_{T}$ from PN functions were presented and their parameters were also determined in Theorems \ref{main-51} and \ref{main-52}, where $p$ is an odd prime.
\end{itemize}
Furthermore, the parameters of the shortened codes look new.

In addition to the works in \cite{LDT20} and this paper, other linear codes with good parameters may be produced with the shortening technique. However, it seems hard to determine the weight distributions of shortened and punctured codes in general. The reader is cordially invited to join the adventure in this direction.

\end{document}